\def\year{2018}\relax
\documentclass[letterpaper]{article}
\usepackage{aaai18}
\usepackage{times}
\usepackage{helvet}
\usepackage{courier}
\usepackage{url}
\usepackage{graphicx}
\usepackage{macros}
\usepackage{thm-restate}
\usepackage{tikz}
\usepackage{tabularx}
\usepackage{multirow}
\usepackage{etoolbox}
\newtoggle{withappendix}
\toggletrue{withappendix}

\begin{document}

\title{Optimised Maintenance of Datalog Materialisations}
\author{
    Pan Hu \and Boris Motik \and Ian Horrocks \\ 
	Department of Computer Science, University of Oxford \\
	Oxford, United Kingdom \\
	firstname.lastname@cs.ox.ac.uk
}

\maketitle

\begin{abstract}
To efficiently answer queries, datalog systems often materialise all
consequences of a datalog program, so the materialisation must be updated
whenever the input facts change. Several solutions to the materialisation
update problem have been proposed. The \emph{Delete/Rederive} (DRed) and the
\emph{Backward/Forward} (B/F) algorithms solve this problem for general
datalog, but both contain steps that evaluate rules `backwards' by matching
their heads to a fact and evaluating the partially instantiated rule bodies as
queries. We show that this can be a considerable source of overhead even on
very small updates. In contrast, the \emph{Counting} algorithm does not
evaluate the rules `backwards', but it can handle only nonrecursive rules. We
present two hybrid approaches that combine DRed and B/F with Counting so as to
reduce or even eliminate `backward' rule evaluation while still handling
arbitrary datalog programs. We show empirically that our hybrid algorithms are
usually significantly faster than existing approaches, sometimes by orders of
magnitude.
 
\end{abstract}

\section{Introduction}\label{sec:introduction}

Datalog \cite{abiteboul1995foundations} is a rule language that is widely used
in modern information systems. Datalog rules can declaratively specify tasks in
data analysis applications \cite{luteberget2016rule,piro2016semantic}, allowing
application developers to focus on the objective of the analysis---that is, on
specifying \emph{what} needs to be computed rather than \emph{how} to compute
it \cite{markl2014breaking}. Datalog can also capture OWL 2 RL
\cite{motik2009owl} ontologies possibly extended with SWRL rules
\cite{horrocks2004swrl}. It is implemented in systems such as WebPIE
\cite{urbani2012webpie}, VLog \cite{urbani2016column}, Oracle's RDF Store
\cite{wu2008implementing}, OWLIM \cite{bishop2011owlim}, and RDFox
\cite{nenov2015rdfox}, and it is extensively used in practice.

When performance is critical, datalog systems usually precompute the
\emph{materialisation} (i.e., the set of all consequences of a program and the
explicit facts) in a preprocessing step so that all queries can later be
evaluated directly over the materialisation. Recomputing the materialisation
from scratch whenever the explicit facts change can be expensive. Systems thus
typically use an \emph{incremental maintenance algorithm}, which aims to avoid
repeating most of the work. Fact insertion can be effectively handled using the
semina\"ive algorithm \cite{abiteboul1995foundations}, but deletion is much
more involved since one has to check whether deleted facts have derivations
that persist after the update. The \emph{Delete/Rederive} (DRed) algorithm
\cite{gupta1993maintaining,staudt1996incremental}, the \emph{Backward/Forward}
(B/F) algorithm \cite{motik2015incremental}, and the \emph{Counting} algorithm
\cite{gupta1993maintaining} are well-known solutions to this problem.

The DRed algorithm handles deletion by first overdeleting all facts that depend
on the removed explicit facts, and then rederiving the facts that still hold
after overdeletion. The rederivation stage further involves rederiving all
overdeleted facts that have alternative derivations, and then recomputing the
consequences of the rederived facts until a fixpoint is reached. The algorithm
and its variants have been extensively used in practice
\cite{DBLP:conf/semweb/UrbaniMJHB13,DBLP:conf/cikm/RenP11}. In contrast to
DRed, the B/F algorithm searches for alternative derivations immediately
(rather than after overdeletion) using a combination of backward and forward
chaining. This makes deletion exact and avoids the potential inefficiency of
overdeletion. In practice, B/F often, but not always, outperforms DRed
\cite{motik2015incremental}.

Both DRed and B/F search for derivations of deleted facts by evaluating rules
`backwards': for each rule whose head matches the fact being deleted, they
evaluate the partially instantiated rule body as a query; each query answer
thus corresponds to a derivation. This has two consequences. First, one can
examine rule instances that fire both before and after the update, which is
redundant. Second, evaluating rules `backwards' can be inherently more
difficult than matching the rules during initial materialisation: our
experiments show that this step can, in some cases, prevent effective
materialisation maintenance even for very small updates.

In contrast, the Counting algorithm \cite{gupta1993maintaining} does not
evaluate rules `backwards', but instead tracks the number of distinct
derivations of each fact: a counter is incremented when a new derivation for
the fact is found, and it is decremented when a derivation no longer holds. A
fact can thus be deleted when its counter drops to zero, without the
potentially costly `backward' rule evaluation. The algorithm can also be made
optimal in the sense that it considers precisely the rule instances that no
longer fire after the update and the rule instances that only fire after the
update. The main drawback of Counting is that, unlike DRed and B/F, it is
applicable only to nonrecursive rules \cite{nicolas1983outline}. Recursion is a
key feature of datalog, allowing one to express common properties such as
transitivity. Thus, despite its favourable properties, the Counting algorithm
does not provide us with a general solution to the materialisation maintenance
problem.

Towards the goal of developing efficient general-purpose maintenance
algorithms, in this paper we present two hybrid approaches that combine DRed
and B/F with Counting. The former tracks the nonrecursive and the recursive
derivations separately, which allows the algorithm to eliminate all `backward'
rule evaluation and also limit overdeletion. The latter tracks nonrecursive
derivations only, which eliminates `backward' rule evaluation for nonrecursive
rules; however, recursive rules can still be evaluated `backwards' to eagerly
identify alternative derivations. Both combinations can handle recursive rules,
and they exhibit `pay-as-you-go' behaviour in the sense that they become
equivalent to Counting on nonrecursive rules. Apart from the modest cost of
maintaining counters, our algorithms never involve more computation steps than
their unoptimised counterparts. Thus, our algorithms combine the best aspects
of DRed, B/F, and Counting: without incurring a significant cost, they
eliminate or reduce `backward' rule evaluation, are optimal for nonrecursive
rules, and can also handle recursive rules.

We have implemented our hybrid algorithms and have compared them with the
original DRed and B/F algorithms on several synthetic and real-life benchmarks.
Our experiments show that the cost of counter maintenance is negligible, and
that our hybrid algorithms typically outperform existing solutions, sometimes
by orders of magnitude. Our test system and datasets are available
online.\footnote{\url{http://krr-nas.cs.ox.ac.uk/2017/counting/}}

\section{Preliminaries}\label{sec:preliminaries}

We now introduce datalog with stratified negation. We fix countable, disjoint
sets of \emph{constants} and \emph{variables}. A \emph{term} is a constant or a
variable. A vector of terms is written $\vec t$, and we often treat it as a
set. A (\emph{positive}) \emph{atom} has the form $P(t_1,\dots,t_k)$ where $P$
is a $k$-ary \emph{predicate} and each $t_i$, ${1 \leq i \leq k}$, is a term. A
term or an atom is \emph{ground} if it does not contain variables. A
\emph{fact} is a ground atom, and a \emph{dataset} is a finite set of facts. A
rule $r$ has the form
\begin{displaymath}
    B_1 \wedge \dots \wedge B_m \wedge \stnot B_{m+1} \wedge \dots \wedge \stnot B_n \rightarrow H,
\end{displaymath}
where ${m \geq 0}$, ${n \geq 0}$, and $B_i$ and $H$ are atoms. The \emph{head}
$\head{r}$ of $r$ is the atom $H$, the \emph{positive body} $\pbody{r}$ of $r$
is the set of atoms ${B_1, \dots, B_m}$, and the \emph{negative body}
$\nbody{r}$ of $r$ is the set of atoms ${B_{m+1}, \dots, B_n}$. Rule $r$ must
be \emph{safe}: each variable occurring in $r$ must occur in a positive body
atom.

A \emph{substitution} $\sigma$ is a mapping of finitely many variables to
constants. For $\alpha$ a term, literal, rule, conjunction, or a vector or set
thereof, $\alpha\sigma$ is the result of replacing each occurrence of a
variable $x$ in $\alpha$ with $\sigma(x)$ (if the latter is defined).

A \emph{stratification} $\lambda$ of a program $\Pi$ maps each predicate of
$\Pi$ to a positive integer such that, for each rule ${r \in \Pi}$ with
${\head{r} = P(\vec t)}$, (i)~${\lambda(P) \geq \lambda(R)}$ for each atom
${R(\vec s) \in \pbody{r}}$, and (ii)~${\lambda(P) > \lambda(R)}$ for each atom
${R(\vec s) \in \nbody{r}}$. Program $\Pi$ is \emph{stratifiable} if a
stratification $\lambda$ of $\Pi$ exists. A rule $r$ with ${\head{r} = P(\vec
t)}$ is \emph{recursive} w.r.t.\ $\lambda$ if an atom ${R(\vec s) \in
\pbody{r}}$ exists such that ${\lambda(P) = \lambda(R)}$; otherwise, $r$ is
\emph{nonrecursive} w.r.t.\ $\lambda$. For each positive integer $s$, program
${\strat{\Pi}{s} = \{ r \in \Pi \mid \lambda(\head{r}) = s\}}$ is the
\emph{stratum} $s$ of $\Pi$, and programs $\rstrat{\Pi}{s}$ and
$\nrstrat{\Pi}{s}$ are the recursive and the nonrecursive subsets,
respectively, of $\strat{\Pi}{s}$. Finally, $\Out{s}$ is the set of all facts
that belong to stratum $s$---that is, ${\Out{s} = \{ P(\vec c) \mid \lambda(P)
= s \}}$.

Rule $r'$ is an \emph{instance} of a rule $r$ if a substitution $\sigma$ exists
mapping all variables of $r$ to constants such that ${r' = r\sigma}$. For $I$ a
dataset, the set $\ins{r}{I}$ of instances of $r$ obtained by applying a rule
$r$ to $I$, and the set $\apply{\Pi}{I}$ of facts obtained by applying a
program $\Pi$ to $I$ are defined as follows.
\begin{align}
    \ins{r}{I}      & = \{ r\sigma \mid \pbody{r}\sigma \subseteq I \text{ and } \nbody{r}\sigma \cap I = \emptyset \} \\
    \apply{\Pi}{I}  & = \bigcup_{r \in \Pi} \{ \head{r'} \mid r' \in \ins{r}{I} \}
\end{align}
We often say that each instance in $\ins{r}{I}$ \emph{fires} on $I$. We are now
ready to define the semantics of stratified datalog. Given a dataset $E$ of
\emph{explicit facts} and a stratification $\lambda$ of $\Pi$ with maximum
stratum index number $S$, we define the following sequence of datasets: let
${I^{0}_{\infty} = E}$; let ${I^s_0 = I^{s-1}_{\infty}}$ for index $s$ with ${1
\leq s \leq S}$; let ${I^s_i = I^s_{i-1} \cup \apply{\Pi^s}{I^{s}_{i-1}}}$ for
each integer $i > 0$; and let ${I^s_{\infty} = \bigcup_{i \geq 0}{I^s_i}}$. Set
$I^S_{\infty}$ is called the \emph{materialisation} of $\Pi$ w.r.t.\ $E$ and
$\lambda$. It is well known that $I^S_{\infty}$ does not depend on $\lambda$,
so we usually write it as $\mat{\Pi}{E}$. In this paper, we consider the
problem of maintaining $\mat{\Pi}{E}$: given $\mat{\Pi}{E}$ and datasets $E^-$
and $E^+$, our algorithm computes ${\mat{\Pi}{(E \setminus E^-) \cup E^+}}$
incrementally while minimising the amount of work.

\section{Motivation and Intuition}\label{sec:motivation}

As motivation for our work, we next discuss how evaluating rules `backwards'
can be a significant source of inefficiency during materialisation maintenance.
We base our discussion on the DRed algorithm for simplicity, but our
conclusions apply to the B/F algorithm as well.

\subsection{The DRed Algorithm}

To make our discussion precise, we first present the DRed algorithm
\cite{gupta1993maintaining,staudt1996incremental}. Let $\Pi$ be a program with
a stratification $\lambda$, let $E$ be a set of explicit facts, and assume that
the materialisation ${I = \mat{\Pi}{E}}$ of $\Pi$ w.r.t.\ $E$ has been
computed. Moreover, assume that $E$ should be updated by deleting $E^-$ and
inserting $E^+$. The DRed algorithm efficiently modifies the `old'
materialisation $I$ to the `new' materialisation ${I' = \mat{\Pi}{(E \setminus
E^-) \cup E^+}}$ by deleting some facts and adding others; we call such facts
\emph{affected} by the update.

Due to the update, some rule instances that fire on $I$ will no longer fire on
$I'$, and some rule instances that do not fire on $I$ will fire on $I'$; we
also call such rule instances \emph{affected} by the update. A key problem in
materialisation maintenance is to identify the affected rule instances.
Clearly, the body of each affected rule instance must contain an affected fact.
Based on this observation, the affected rule instances can be efficiently
identified by the following generalisation of the operators $\ins{r}{I}$ and
$\apply{\Pi}{I}$ from Section~\ref{sec:preliminaries}. In particular, let
$\ipos$, $\ineg$, $P$, and $N$ be datasets such that ${P \subseteq \ipos}$ and
${N \cap \ineg = \emptyset}$; then, let
\begin{align}
    &\begin{array}{@{}l@{\;}l@{}}
        \multicolumn{2}{@{}l@{}}{\ins{r}{\ipos, \ineg \appargs P, N} \; =} \\
        \quad \{ r\sigma \mid    & \pbody{r}\sigma \subseteq \ipos \text{ and } \nbody{r}\sigma \cap \ineg = \emptyset, \text{ and} \\
                                 & \pbody{r}\sigma \cap P \neq \emptyset \text{ or } \nbody{r}\sigma \cap N \neq \emptyset \}
    \end{array} \label{eq:ruleinstancepn}
\end{align}
and let
\begin{displaymath}
    \apply{\Pi}{\ipos, \ineg \appargs P, N} = \bigcup_{r \in \Pi} \{ \head{r'} \mid r' \in \ins{r}{\ipos, \ineg \appargs P, N} \}.
\end{displaymath}
Intuitively, the positive and the negative rule atoms are evaluated in $\ipos$
and $\ineg$; sets $P$ and $N$ identify the affected positive and negative
facts; ${\ins{r}{\ipos, \ineg \appargs P, N}}$ are the affected instances of
$r$; and ${\apply{\Pi}{\ipos, \ineg \appargs P, N}}$ are the affected
consequences of $\Pi$. We define ${\ins{r}{\ipos, \ineg}}$ and
${\apply{\Pi}{\ipos, \ineg}}$ analogously to above, but without the condition
`${\pbody{r}\sigma \cap P \neq \emptyset}$ or ${\nbody{r}\sigma \cap N \neq
\emptyset}$'. We omit for readability $\ineg$ whenever ${\ipos = \ineg}$, and
furthermore we omit $N$ when ${N = \emptyset}$. Sets ${\apply{\Pi}{\ipos,
\ineg}}$ and ${\apply{\Pi}{\ipos, \ineg \appargs P, N}}$ can be computed
efficiently in practice by evaluating the body of each rule ${r \in \Pi}$ as a
conjunctive query and instantiating the head as needed.

Algorithm~\ref{alg:dred} formalises DRed. The algorithm processes each stratum
$s$ and accumulates the necessary changes to $I$ in the set $D$ of
\emph{overdeleted} and the set $A$ of \emph{added} facts. The materialisation
is updated in line~\ref{alg:dred:update-I}, so, prior to that, $I$ and ${(I
\setminus D) \cup A}$ are the `old' and the `new' materialisation,
respectively. The computation proceeds in three phases.

In the \emph{overdeletion} phase, $D$ is extended with all facts that depend on
a deleted fact. In line~\ref{alg:dred:del:ND-init} the algorithm identifies the
facts that are explicitly deleted (${E^- \cap \Out{s}}$) or are affected by
deletions in the previous strata (${\apply{\strat{\Pi}{s}}{I \appargs D
\setminus A, A \setminus D}}$), and then in
lines~\ref{alg:dred:del:loop}--\ref{alg:dred:del:endloop} it computes their
consequences. It uses a form of the semina\"ive strategy, which ensures that
each rule instance is considered only once during overdeletion.

In the \emph{one-step rederivation} phase, $R$ is computed as the set of facts
that have been overdeleted, but that hold nonetheless. To this end, in
line~\ref{alg:dred:OneStepRederive} the algorithm considers each fact $F$ in
${D \cap \Out{s}}$, and it adds $F$ to $R$ if $F$ is explicit or it is
rederived by a rule instance. The latter involves evaluating rules `backwards':
the algorithm identifies each rule ${r \in \strat{\Pi}{s}}$ whose head can be
matched to $F$, and it evaluates over the `new' materialisation the body of $r$
as a query with the head variables bound; fact $F$ holds if the query returns
at least one answer. As we discuss shortly, this step can be a major source of
inefficiency in practice, and the main contribution of this paper is
eliminating `backward' rule evaluation and thus significantly improving the
performance.

In the \emph{insertion} step, in line~\ref{alg:dred:ins:NA-init} the algorithm
combines the one-step rederived facts ($R$) with the explicitly added facts
(${E^+ \cup \Out{s}}$) and the facts added due to the changes in the previous
strata ($\apply{\strat{\Pi}{s}}{(I \setminus D) \cup A \appargs A \setminus D,
D \setminus A}$), and then in
lines~\ref{alg:dred:ins:loop}--\ref{alg:dred:ins:loop:end} it computes all of
their consequences and adds them to $A$. Again, the semina\"ive strategy
ensures that each rule instance is considered only once during insertion.

\begin{algorithm}[t]
\fontsize{8}{10}\selectfont
\caption{\textsc{DRed}$(\Pi, \lambda, E, I, E^-, E^+)$} \label{alg:dred}
\begin{algorithmiccont}
    \State $D \defeq A \defeq \emptyset$, \quad $E^- = (E^- \cap E) \setminus E^+$, \quad $E^+ = E^+ \setminus E$                               \label{alg:dred:init}
    \For{each stratum index $s$ with ${1 \leq s \leq S}$}                                                                                       \label{alg:dred:stratum-loop}
        \State \Call{Overdelete}{}                                                                                                              \label{alg:dred:Overdelete}
        \State $R \defeq \{ F \in D \cap \Out{s} \mid F \in E \setminus E^- \text{ or there exist } r \in \strat{\Pi}{s}$ and
        \Statex \hspace{4em} $r' \in \ins{r}{I \setminus (D \setminus A), I \cup A} \text{ with } F = \head{r'} \}$                             \label{alg:dred:OneStepRederive}
        \State \Call{Insert}{}                                                                                                                  \label{alg:dred:Insert}
    \EndFor                                                                                                                                     \label{alg:dred:stratum-loop:end}
    \State $E \defeq (E \setminus E^-) \cup E^+$, \quad $I \defeq (I \setminus D) \cup A$                                                       \label{alg:dred:update-I}
    \Statex
    \Procedure{Overdelete}{}
        \State $N_D \defeq (E^- \cap \Out{s}) \cup \apply{\strat{\Pi}{s}}{I \appargs D \setminus A, A \setminus D}$                             \label{alg:dred:del:ND-init}
        \Loop                                                                                                                                   \label{alg:dred:del:loop}
            \State $\Delta_D \defeq N_D \setminus D$                                                                                            \label{alg:dred:del:DeltaD}
            \If{$\Delta_D = \emptyset$}                                                                                                         \label{alg:dred:del:if}
                \Break
            \EndIf
            \State $N_D \defeq \apply{\rstrat{\Pi}{s}}{I \setminus (D \setminus A), I \cup A \appargs \Delta_D}$                                \label{alg:dred:del:ND-update}
            \State $D \defeq D \cup \Delta_D$                                                                                                   \label{alg:dred:del:D-update}
        \EndLoop                                                                                                                                \label{alg:dred:del:endloop}
    \EndProcedure
    \Statex
    \Procedure{Insert}{}
        \State $N_A \defeq R \cup (E^+ \cap \Out{s}) \cup \apply{\strat{\Pi}{s}}{(I \setminus D) \cup A \appargs A \setminus D, D \setminus A}$ \label{alg:dred:ins:NA-init}
        \Loop                                                                                                                                   \label{alg:dred:ins:loop}
            \State $\Delta_A \defeq N_A \setminus ((I \setminus D) \cup A)$                                                                     \label{alg:dred:ins:DeltaA}
            \If{$\Delta_A = \emptyset$}
                \Break
            \EndIf
            \State $A \defeq A \cup \Delta_A$                                                                                                   \label{alg:dred:ins:A-update}
            \State $N_A \defeq \apply{\rstrat{\Pi}{s}}{(I \setminus D) \cup A \appargs \Delta_A}$                                               \label{alg:dred:ins:NA-update}
        \EndLoop                                                                                                                                \label{alg:dred:ins:loop:end}
    \EndProcedure
\end{algorithmiccont}
\end{algorithm}

\subsection{Problems with Evaluating Rules `Backwards'}

The one-step rederivation in line~\ref{alg:dred:OneStepRederive} of
Algorithm~\ref{alg:dred} evaluates rules `backwards'. In this section we
present two examples that demonstrate how this can be a major source of
inefficiency. Both examples are derived from datasets we used in our empirical
evaluation that we present in Section~\ref{sec:evaluation}; hence, these
problems actually arise in practice.

Our discussion depends on several details. In particular, we assume that all
facts are indexed so that all facts matching any given atom (possibly
containing constants) can be identified efficiently. Moreover, we assume that
conjunctive queries corresponding to rule bodies are evaluated left-to-right:
for each match of the first conjunct, we partially instantiate the rest of the
body and match it recursively. Finally, we assume that query atoms are
reordered prior to evaluation to obtain an efficient evaluation plan.

\begin{example}\label{ex1}
Let $\Pi$ and $E$ be the program and the dataset as specified in
\eqref{inefficientbackwardrule} and \eqref{inefficientbackwarddataset},
respectively.
\begin{align}
    & R(x,y_1) \wedge R(x,y_2) \rightarrow S(y_1,y_2)       \label{inefficientbackwardrule} \\
    & E = \{ R(a_i,b),R(a_i,c_i) \mid 1 \leq i \leq n \}    \label{inefficientbackwarddataset}
\end{align}
The materialisation $\mat{\Pi}{E}$ consists of $E$ extended with facts
$S(b,b)$, $S(b,c_i)$, $S(c_i,b)$, and $S(c_i,c_i)$ for ${1 \leq i \leq n}$.

During materialisation, the body of rule \eqref{inefficientbackwardrule} can be
evaluated efficiently left-to-right: we match $R(x,y_1)$ to either $R(a_i,b)$
or $R(a_i,c_i)$; this instantiates $R(x,y_2)$ as $R(a_i,y_2)$, and we use the
index to find the matching facts $R(a_i,b)$ and $R(a_i,c_i)$. Thus, $R(x,y_1)$
has $2n$ matches, each of which contributes to two matches of $R(x,y_2)$, so
the overall cost of rule matching is $O(n)$. The rule body is symmetric, so
reordering the body atoms has no effect.

Now assume that we delete all $R(a_i,c_i)$ with ${1 \leq i \leq n}$. DRed then
overdeletes all $S(b,c_i)$, $S(c_i,b)$ and $S(c_i,c_i)$ facts in
lines~\ref{alg:dred:del:ND-init}--\ref{alg:dred:del:endloop}, and this can be
done efficiently as in the previous paragraph. Next, in one-step rederivation,
the algorithm will match these facts to the head of the rule
\eqref{inefficientbackwardrule} and obtain queries ${R(x,b) \wedge R(x,c_i)}$,
${R(x,c_i) \wedge R(x,b)}$, and ${R(x,c_i) \wedge R(x,c_i)}$. All but the last
of these queries contain atom $R(x,b)$ and, no matter how we reorder the body
atoms of \eqref{inefficientbackwardrule}, we have $n$ queries where $R(x,b)$ is
evaluated first. Each of these $n$ queries identifies $n$ candidate matches
$R(a_i,b)$ using the index only to find out that the second atom cannot be
matched. Thus, $R(x,b)$ is matched to $n^2$ facts in total, so the cost of
one-step rederivation is $O(n^2)$---one degree higher than for materialisation.
\end{example}

Example~\ref{ex1} shows that evaluating a rule `backwards' can be inherently
more difficult than evaluating it during materialisation, thus giving rise to a
dominating source of inefficiency. In fact, evaluating a rule with $m$ body
atoms `backwards' can be seen as answering a query with $m+1$ atoms, where the
head of the rule is an extra query atom; since the number of atoms determines
the complexity of query evaluation, this extra atom increases the algorithm's
complexity.

Our next example shows that this problem is exacerbated if the space of
admissible plans for queries corresponding to rule bodies is further
restricted. This is common in systems that provide \emph{built-in functions}.
In particular, to facilitate manipulation of concrete values such as strings or
integers, datalog systems often allow rule bodies to contain \emph{built-in
atoms} of the form ${(t \defeq \mathit{exp})}$, where $t$ is a term and
$\mathit{exp}$ is an expression constructed using constants, variables,
functions, and operators as usual. For example, a built-in atom can have the
form ${(z \defeq z_1 + z_2)}$, and it assigns to $z$ the sum of $z_1$ and
$z_2$. The set of supported functions vary among implementations, but a common
feature is that all values in $\mathit{exp}$ must be bound by prior atoms
before the built-in atom can be evaluated. As we show next, this can be
problematic.

\begin{example}\label{ex2}
Let program $\Pi$ consist of rules \eqref{ex2rule1} and \eqref{ex2rule2}. If we
read $B(s,t,n)$ as saying that there is an edge from node $s$ to node $t$ of
length $n$, then the program entails $D(s,n)$ if there exists a path of length
$n$ from node $a$ to node $s$.
\begin{align}
    B(a,y,z)                                                & \rightarrow D(y,z)    \label{ex2rule1} \\
    D(x,z_1) \wedge B(x,y,z_2) \wedge (z \defeq z_1 + z_2)  & \rightarrow D(y,z)    \label{ex2rule2}
\end{align}
Let $E$ be the dataset as specified below.
\begin{align*}
    E = \{ B(a,b_1,1), B(a,c_i,1), B(b_i,d_j,1) \mid 1 \leq i,j \leq n \}
\end{align*}

During materialisation, rule \eqref{ex2rule1} first derives $D(b_1,1)$ and all
$D(c_i,1)$ with ${1 \leq i \leq n}$, so the cost of this step is $O(n)$. Next,
atom $D(x,z_1)$ in rule \eqref{ex2rule2} is matched to $n$ facts $D(c_i,1)$
without deriving anything. Atom $D(x,z_1)$ is also matched to $D(b_1,1)$ once,
so atom $B(x,y,z_2)$ is instantiated to $B(b_1,y,z_2)$ and matched to $n$ facts
$B(b_1,d_j,1)$, deriving $n$ facts $D(d_j,2)$. Thus, the cost of rule matching
is $O(n)$.

Now assume that $B(a,b_1,1)$ is deleted. Then, $D(b_1,1)$ and all $D(d_j,2)$
can be efficiently overdeleted as in the previous paragraph, but trying to
prove them is much more difficult. Matching each $D(d_j,2)$ to the head of
\eqref{ex2rule1} produces a query $B(a,d_j,2)$, which does not produce a rule
instance. Moreover, matching $D(d_j,2)$ to the head of \eqref{ex2rule2}
produces a query ${D(x,z_1) \wedge B(x,d_j,z_2) \wedge (2 \defeq z_1+z_2)}$.
Now, as we discussed earlier, $z_1$ and $z_2$ must both be bound before we can
evaluate the built-in atom ${(2 \defeq z_1+z_2)}$. If we evaluate
$B(x,d_j,z_2)$ first, then we try $n$ facts $B(b_i,d_j,1)$ with ${1 \leq i \leq
n}$; for each of them, atom $D(x,z_1)$ is instantiated as $D(b_i,z_1)$ and is
not matched in the surviving facts. In contrast, if we evaluate $D(x,z_1)$
first, then we try $n$ facts $D(c_i,1)$; for each of them, atom $B(x,d_j,z_2)$
is instantiated as $B(c_i,y,z_2)$ and is not matched. Thus, regardless of how
we reorder the body of \eqref{ex2rule2}, the first atom considers a total of
$n^2$ facts, so the cost of one-step rederivation is $O(n^2)$.

To overcome this, one might rewrite the built-in atom as ${(z_1 \defeq z -
z_1)}$ or ${(z_2 \defeq z - z_2)}$ so that it can be evaluated immediately
after $z$ and either $z_1$ or $z_2$ are bound. Either way, one-step
rederivation still takes $O(n^2)$ steps on our example. Also, built-in
expressions are often not invertible.
\end{example}

\section{Combining DRed with Counting}\label{sec:dredcnt}

We now address the inefficiencies we outlined in Section~\ref{sec:motivation}.
Towards this goal, in Section~\ref{sec:dredcnt:intuition} we first present the
intuitions, and then in Section~\ref{sec:dredcnt:formalisation} we formalise
our solution.

\subsection{Intuition}\label{sec:dredcnt:intuition}

As we already mentioned in Section~\ref{sec:introduction}, the Counting
algorithm \cite{gupta1993maintaining} does not evaluate rules `backwards';
instead, it tracks the number of derivations of each fact. The main drawback of
Counting is that it cannot handle recursive rules. We now illustrate the
intuition behind our DRed$^c$ algorithm, which combines DRed with Counting in a
way that eliminates `backward' rule evaluation, while still supporting
recursive rules.

The DRed$^c$ algorithm associates with each fact two counters that track the
derivations via the nonrecursive and the recursive rules separately. The
counters are decremented (resp.\ incremented) when the associated fact is
derived in overdeletion (resp.\ insertion), which allows for two important
optimisations. First, as in the Counting algorithm, the nonrecursive counter
always reflects the number of derivations from facts in earlier strata; hence,
a fact with a nonzero nonrecursive counter should never be overdeleted because
it clearly remains true after the update. This optimisation captures the
behaviour of Counting on nonrecursive rules and it also helps limit
overdeletion. Second, if we never overdelete facts with nonzero nonrecursive
counters, the only way for a fact to still hold after overdeletion is if its
recursive counter is nonzero; hence, we can replace `backward' rule evaluation
by a simple check of the recursive counter. Note, however, that the recursive
counters can be checked only after overdeletion finishes. This optimisation
extends the idea of Counting to recursive rules to completely avoid `backward'
rule evaluation. The following example illustrates these ideas and compares
them to DRed.

\begin{example}\label{ex3}
Let $\Pi$ be the program containing rule~\eqref{ex3rule1}.
\begin{align}
    A(x) \wedge B(x,y) \rightarrow A(y) \label{ex3rule1}
\end{align}
Moreover, let $E$ be defined as follows:
\begin{align*}
     E = \{ A(a), A(b), A(d), B(a,c), B(b,c), B(c,d), B(d,e) \}
\end{align*}
The materialisation $\mat{\Pi}{E}$ extends $E$ with $A(c)$ and $A(e)$.
Figure~\ref{fig:ex3} shows the dependencies between derivations using arrows.
For clarity, we do not show the $B$-facts.

Now assume that $A(a)$ is deleted. The standard DRed algorithm first
overdeletes $A(a)$, $A(c)$, $A(d)$, and $A(e)$; it rederives $A(d)$ since the fact 
is in $E \setminus E^-$; it rederives $A(c)$ by evaluating rule~\eqref{ex3rule1} 
`backwards'; and it derives $A(d)$ and $A(e)$ from the rederived facts.

Now consider applying the DRed$^c$ to the same update. For each fact,
Figure~\ref{fig:ex3} shows a pair consisting of the nonrecursive and the
recursive counter before the update (row T1), after overdeletion (row T2), and
after the update (row T3). Note that the presence of a fact in $E$ is akin to a
nonrecursive derivation, so facts $A(a)$, $A(b)$, and $A(d)$ have nonrecursive
derivation counts of one before the update. Now $A(c)$ is derived from $A(a)$
and $A(b)$ using the recursive rule \eqref{ex3rule1}, so the recursive counter
for $A(c)$ is two. Analogously, $A(d)$ and $A(e)$ have just one recursive
derivation each. During overdeletion, $A(a)$ is first removed from $E$, so the
nonrecursive counter of $A(a)$ is decremented to zero and the fact is deleted.
Since $A(a)$ derives $A(c)$ via rule \eqref{ex3rule1}, the recursive counter of
$A(c)$ is decremented; since the nonrecursive counter of $A(c)$ is zero, the
fact is overdeleted. Since $A(c)$ derives $A(d)$ via rule \eqref{ex3rule1}, the
recursive counter of $A(d)$ is decremented. Now the nonrecursive counter of
$A(d)$ is nonzero, so we know that $A(d)$ holds after the update; hence, the
fact is \emph{not} overdeleted, and the overdeletion phase stops. Thus, while
DRed overdeletes four facts, DRed$^c$ overdeletes only $A(a)$ and $A(c)$, and
does not `touch' $A(e)$.

Next, DRed$^c$ proceeds to one-step rederivation. The recursive counter of
$A(c)$ is nonzero, which means that the fact has a recursive derivation (from
$A(b)$ in this case) that is not affected. Thus, DRed$^c$ rederives $A(c)$
without any `backward' rule evaluation.

Finally, DRed$^c$ applies insertion. Since $A(c)$ derives $A(d)$ via
\eqref{ex3rule1}, the recursive counter of $A(d)$ is incremented. Fact $A(d)$,
however, was not overdeleted, so insertion stops.
\end{example}

\begin{figure}[tb]
\centering
\begin{tikzpicture}
    \tikzstyle{derived}=[font=\boldmath, line width=1pt]

    \node (at1) at (0,2.2)      {T1: (1,0)};
    \node (at2) at (0,1.8)      {T2: (0,0)};
    \node (at3) at (0,1.4)      {T3: (0,0)};
    \node (a) at (0,1)          {$A(a)$};
    \node (b) at (0,-1)         {$A(b)$};
    \node (bt1) at (0,-1.4)     {T1: (1,0)};
    \node (bt2) at (0,-1.8)     {T2: (1,0)};
    \node (bt3) at (0,-2.2)     {T3: (1,0)};
    \node (ct1) at (2,1.2)      {T1: (0,2)};
    \node (ct2) at (2,0.8)      {T2: (0,1)};
    \node (ct3) at (2,0.4)      {T3: (0,1)};
    \node (c) at (2,0)          {$A(c)$};
    \node (dt1) at (4,1.2)      {T1: (1,1)};
    \node (dt2) at (4,0.8)      {T2: (1,0)};
    \node (dt3) at (4,0.4)      {T3: (1,1)};
    \node (d) at (4,0)          {$A(d)$};
    \node (e) at (4,-1)         {};
    \node (et1) at (6,1.2)      {T1: (0,1)};
    \node (et2) at (6,0.8)      {T2: (0,1)};
    \node (et3) at (6,0.4)      {T3: (0,1)};
    \node (f) at (6,0)          {$A(e)$};
    \node (g) at (-1.2,1)       {};
    \node (h) at (-1.2,-1)      {};

    \draw[-latex] (a) to node[below] {} (c);
    \draw[-latex] (b) to node[below] {} (c);
    \draw[-latex] (c) to node[below] {} (d);
    \draw[-latex] (d) to node[below] {} (f);
    \draw[-latex] (e) to node[below] {} (d);
    \draw[-latex] (g) to node[below] {} (a);
    \draw[-latex] (h) to node[below] {} (b);

\end{tikzpicture}
\caption{Derivations for Example~\ref{ex3}}\label{fig:ex3}
\end{figure}
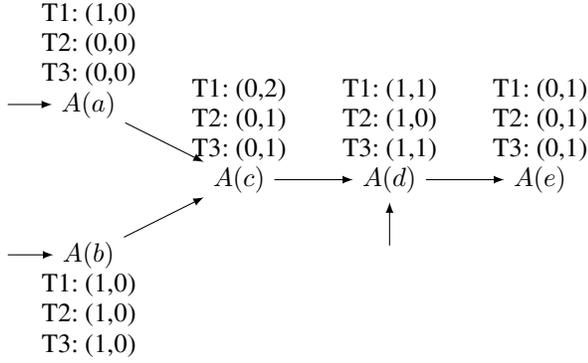

By avoiding `backward' rule evaluation, DRed$^c$ removes the dominating source
of inefficiency on Examples~\ref{ex1} and~\ref{ex2}. In fact, on the
nonrecursive program from Example~\ref{ex1}, the recursive counter is never
used and DRed$^c$ performs the same inferences as the Counting algorithm.

\subsection{Formalisation}\label{sec:dredcnt:formalisation}

We now formalise our DRed$^c$ algorithm. Our definitions use the standard
notion of \emph{multisets}---a generalisation of sets where each element is
associated with a positive integer called the \emph{multiplicity} specifying
the number of the element's occurrences in the multiset. Moreover, $\oplus$ is
the multiset union operator, which adds the elements' multiplicities. If an
operand of $\oplus$ is a set, it is treated as a multiset where all elements
have multiplicity one. Finally, we extend the notion of rule matching to
correctly reflect the number of times a fact is derived: for $\ipos$, $\ineg$,
$P$, and $N$ datasets with ${P \subseteq \ipos}$ and ${N \cap \ineg =
\emptyset}$, we define ${\mapply{\Pi}{\ipos, \ineg \appargs P, N}}$ as the
multiset containing a distinct occurrence of $\head{r'}$ for each rule ${r \in
\Pi}$ and its instance ${r' \in \ins{r}{\ipos, \ineg \appargs P, N}}$. This
multiset can be computed analogously to ${\apply{\Pi}{\ipos, \ineg \appargs P,
N}}$.

Just like DRed, the DRed$^c$ takes as input a program $\Pi$, a stratification
$\lambda$, a set of explicit facts $E$ and its materialisation ${I =
\mat{\Pi}{E}}$, and the sets of facts $E^-$ and $E^+$ to remove from and add to
$E$. Additionally, the algorithm also takes as input maps $\Cnr$ and $\Cr$ that
associate each fact $F$ with its nonrecursive and recursive counters $\Cnr[F]$
and $\Cr[F]$, respectively. These maps should correctly reflect the relevant
numbers of derivations. Formally, $\Cnr$ and $\Cr$ must be \emph{compatible}
with $\Pi$, $\lambda$, and $E$, which is the case if ${\Cnr[F] = \Cr[F] = 0}$
for each fact ${F \not\in I}$, and, for each fact ${F \in I}$ and $s$ the
stratum index such that ${F \in \Out{s}}$ (i.e., $s$ is the index of the
stratum that $F$ belongs to),
\begin{itemize}
    \item $\Cnr[F]$ is the multiplicity of $F$ in ${E \oplus
    \mapply{\nrstrat{\Pi}{s}}{I}}$, and
        
    \item $\Cr[F]$ is the multiplicity of $F$ in
    ${\mapply{\rstrat{\Pi}{s}}{I}}$.
\end{itemize}
For simplicity, we assumes that $\Cnr$ and $\Cr$ are defined on all facts, and
that ${\Cnr[F] = \Cr[F] = 0}$ holds for ${F \not\in I}$; thus, we can simply
increment the counters for each newly derived fact in procedure
\Call{Insert}{}. In practice, however, one can maintain counters only for the
derived facts and initialise the counters to zero for the freshly derived facts.

DRed$^c$ is formalised in Algorithm~\ref{alg:dredcnt}. Its structure is similar
to DRed, with the following main differences: instead of evaluating rules
`backwards', one-step rederivation simply checks the recursive counters
(line~\ref{alg:dredcnt:OneStepRederive}); a fact is overdeleted only if the
nonrecursive derivation counter is zero (line~\ref{alg:dredcnt:del:DeltaD});
and the derivation counters are decremented in overdeletion
(lines~\ref{alg:dredcnt:del:decnrloop}--\ref{alg:dredcnt:del:decrloop:end}
and~\ref{alg:dredcnt:del:ND-update}--\ref{alg:dredcnt:del:ND-updateend}) and
incremented in insertion
(lines~\ref{alg:dredcnt:ins:incnrloop}--\ref{alg:dredcnt:ins:incrloop:end}
and~\ref{alg:dredcnt:ins:incrsecondloop}--\ref{alg:dredcnt:ins:NA-update}). The
algorithm also accumulates changes to the materialisation in sets $D$ and $A$
by iteratively processing the strata of $\lambda$ in three phases.

\begin{algorithm}[t]
\fontsize{8}{10}\selectfont
\caption{\textsc{DRed$^c$}$(\Pi, \lambda, E, I, E^-, E^+, \Cnr, \Cr)$}\label{alg:dredcnt}
\begin{algorithmiccont}
    \State $D \defeq A \defeq \emptyset$, \quad $E^- = (E^- \cap E) \setminus E^+$, \quad $E^+ = E^+ \setminus E$                       \label{alg:dredcnt:init}
    \For{each stratum index $s$ with ${1 \leq s \leq S}$}                                                                               \label{alg:dredcnt:stratum-loop}
        \State \Call{Overdelete}{}                                                                                                      \label{alg:dredcnt:Overdelete}
        \State $R \defeq \{ F \in D \cap \Out{s} \mid \Cr[F] > 0 \}$                                                                    \label{alg:dredcnt:OneStepRederive}
        \State \Call{Insert}{}                                                                                                          \label{alg:dredcnt:Insert}
    \EndFor                                                                                                                             \label{alg:dredcnt:stratum-loop:end}
    \State $E \defeq (E \setminus E^-) \cup E^+$, \quad $I \defeq (I \setminus D) \cup A$                                               \label{alg:dredcnt:update-I}
    \Statex
    \Procedure{Overdelete}{}
        \State $N_D \defeq \emptyset$
        \For{$F \in (E^- \cap \Out{s}) \oplus \mapply{\nrstrat{\Pi}{s}}{I \appargs D \setminus A, A \setminus D}$}                      \label{alg:dredcnt:del:decnrloop}
            \State $N_D \defeq N_D \cup \{ F \}$, \quad $\Cnr[F] \defeq \Cnr[F] - 1$                                                    \label{alg:dredcnt:del:decnrloop:end}
        \EndFor
        \For{$F \in \mapply{\rstrat{\Pi}{s}}{I \appargs D \setminus A, A \setminus D}$}                                                 \label{alg:dredcnt:del:decrloop}
            \State $N_D \defeq N_D \cup \{ F \}$, \quad $\Cr[F] \defeq \Cr[F] - 1$                                                      \label{alg:dredcnt:del:decrloop:end}
        \EndFor
        \Loop                                                                                                                           \label{alg:dredcnt:del:loop}
            \State $\Delta_D \defeq \{F \in N_D \setminus D \mid \Cnr[F] = 0\}$                                                         \label{alg:dredcnt:del:DeltaD}
            \If{$\Delta_D = \emptyset$}                                                                                                 \label{alg:dredcnt:del:if}
                \Break
            \EndIf
            \For{$F \in \mapply{\rstrat{\Pi}{s}}{I \setminus (D \setminus A), I \cup A \appargs \Delta_D}$}                             \label{alg:dredcnt:del:ND-update}
                \State $N_D \defeq N_D \cup \{ F \}$, \quad $\Cr[F] \defeq \Cr[F] - 1$                                                  \label{alg:dredcnt:del:ND-updateend}
            \EndFor
            \State $D \defeq D \cup \Delta_D$                                                                                           \label{alg:dredcnt:del:D-update}
        \EndLoop                                                                                                                        \label{alg:dredcnt:del:endloop}
    \EndProcedure
    \Statex
    \Procedure{Insert}{}
        \State $N_A \defeq R$                                                                                                           \label{alg:dredcnt:ins:NAequalsR}
        \For{$F \in (E^+ \cap \Out{s}) \oplus \mapply{\nrstrat{\Pi}{s}}{(I \setminus D) \cup A \appargs A \setminus D, D \setminus A}$} \label{alg:dredcnt:ins:incnrloop}
            \State $N_A \defeq N_A \cup \{ F \}$, \quad $\Cnr[F] \defeq \Cnr[F] + 1$                                                    \label{alg:dredcnt:ins:incnrloop:end}
        \EndFor
        \For{$F \in \mapply{\rstrat{\Pi}{s}}{(I \setminus D) \cup A \appargs A \setminus D, D \setminus A}$}                            \label{alg:dredcnt:ins:incrloop}
            \State $N_A \defeq N_A \cup \{F\}$, \quad $\Cr[F] \defeq \Cr[F] + 1$                                                        \label{alg:dredcnt:ins:incrloop:end}
        \EndFor
        \Loop                                                                                                                           \label{alg:dredcnt:ins:loop}
            \State $\Delta_A \defeq N_A \setminus ((I \setminus D) \cup A)$                                                             \label{alg:dredcnt:ins:DeltaA}
            \If{$\Delta_A = \emptyset$}
                \Break
            \EndIf
            \State $A \defeq A \cup \Delta_A$                                                                                           \label{alg:dredcnt:ins:A-update}
            \For{$F \in \mapply{\rstrat{\Pi}{s}}{(I \setminus D) \cup A \appargs \Delta_A}$}                                            \label{alg:dredcnt:ins:incrsecondloop}
                \State $N_A \defeq N_A \cup \{ F \}$, \quad $\Cr[F] \defeq \Cr[F] + 1$
            \EndFor                                                                                                                     \label{alg:dredcnt:ins:NA-update}
        \EndLoop                                                                                                                        \label{alg:dredcnt:ins:loop:end}
    \EndProcedure
\end{algorithmiccont}
\end{algorithm}

In the overdeletion phase, DRed$^c$ first considers explicitly deleted facts or
facts affected by the changes in earlier strata
(lines~\ref{alg:dredcnt:del:decnrloop}--\ref{alg:dredcnt:del:decrloop:end}).
This is analogous to line~\ref{alg:dred:del:ND-init} of DRed, but DRed$^c$ must
distinguish $\nrstrat{\Pi}{s}$ from $\rstrat{\Pi}{s}$ so it can decrement the
appropriate counters. Next, DRed$^c$ identifies the set $\Delta_D$ of facts
that have not yet been deleted and whose nonrecursive counter is zero
(line~\ref{alg:dredcnt:del:DeltaD}): a fact with a nonzero nonrecursive counter
will always be part of the `new' materialisation. Note that recursive
derivations can be cyclic, so we cannot use the recursive counter to further
constrain overdeletion at this point. Then, in
lines~\ref{alg:dredcnt:del:if}--\ref{alg:dredcnt:del:endloop} the algorithm
propagates consequences of $\Delta_D$ just like Algorithm~\ref{alg:dred}, with
additionally decrementing the recursive counters in
line~\ref{alg:dredcnt:del:ND-updateend}.

In the one-step rederivation phase, instead of evaluating rules `backwards',
DRed$^c$ just checks the recursive counter of each fact ${F \in D \cap
\Out{s}}$ (line~\ref{alg:dredcnt:OneStepRederive}): if ${\Cr[F] \neq 0}$, then
some derivations of $F$ were not `touched' by overdeletion so $F$ holds in the
`new' materialisation. Conversely, if ${\Cr[F] = 0}$, then ${F \in D}$
guarantees that ${\Cnr[F] = 0}$ holds as well, so $F$ is not one-step
rederivable by a rule in $\Pi$.

The insertion phase of DRed$^c$ just uses the semina\"ive evaluation while
incrementing the counters appropriately.

Without recursive rules, DRed$^c$ becomes equivalent to Counting, and it is
optimal in the sense that only affected rule instances are considered during
the update. Moreover, the computational complexities of both DRed$^c$ and DRed
are the same as for the semi-naive materialisation algorithm: ExpTime in
combined and PTime in data complexity \cite{dantsin2001complexity}. Finally,
DRed$^c$ never performs more inferences than DRed and is thus more efficient.
Theorem~\ref{theoremcorrectness} shows that our algorithm is correct, and its
proof is given in \iftoggle{withappendix}{full in Appendix~\ref{sec:proof}}{an
extended technical report \cite{extended-version}}.

\begin{restatable}{theorem}{dredcntcorrectness}\label{theoremcorrectness}
    Algorithm~\ref{alg:dredcnt} correctly updates ${I = \mat{\Pi}{E}}$ to ${I'
    = \mat{\Pi}{E'}}$ for ${E' = (E \setminus E^-) \cup E^+}$, and it updates
    $\Cnr$ and $\Cr$ so they are compatible with $\Pi$, $\lambda$, and $E'$.
\end{restatable}

\section{Combining B/F with Counting}\label{sec:bfcnt}

The B/F algorithm by \citet{motik2015incremental} uses a combination of
backward and forward chaining that makes the deletion phase exact. More
specifically, when a fact ${F \in \Delta_D}$ is considered during deletion, the
algorithm uses a combination of backward and forward chaining to look for
alternative derivations of $F$, and it deletes $F$ only if no such derivation
can be found. Backward chaining allows B/F to be much more efficient than DRed
on many datasets, and this is particularly the case if a program contains many
recursive rules. Thus, we cannot hope to remove all `backward' rule evaluation
without eliminating the algorithm's main advantage.

Still, there is room for improvement: backward chaining involves `backward'
evaluation of both nonrecursive and recursive rules, and we can use
nonrecursive counters to eliminate the former. Algorithm~\ref{alg:bfcnt}
formalises B/F$^c$---our combination of the B/F algorithm by
\citet{motik2015incremental} with Counting. The main difference to the original
B/F algorithm is that B/F$^c$ associates with each fact a nonrecursive counter
that is maintained in
lines~\ref{alg:bfcnt:del:decnrcntloop}--\ref{alg:bfcnt:del:decnrcntloop:end}
and~\ref{alg:bfcnt:ins:incnrloop}--\ref{alg:bfcnt:ins:incnrloop:end}, and,
instead of evaluating nonrecursive rules `backwards' to explore alternative
derivations of a fact, it just checks in line~\ref{alg:bfcnt:saturate:base}
whether the nonrecursive counter is nonzero. We know that a fact holds if its
nonrecursive counter is nonzero; otherwise, we apply backward chaining to
\emph{recursive rules only}.We next describe the algorithm's steps in more
detail.

Procedure~\Call{DeleteUnproved}{} plays an analogous role to the overdeletion
step of DRed and DRed$^c$. The procedure maintains the nonrecursive counter for
each fact in the same way as DRed$^c$, and the main difference is that a fact
$F$ is deleted (i.e., added to $\Delta_D$) in line~\ref{alg:bfcnt:del:deltaD}
only if no alternative derivation can be found using a combination of backward
and forward chaining implemented in functions \Call{Check}{} and
\Call{Saturate}{}. If an alternative derivation is found, $F$ is added to the
set $P$ of \emph{proved} facts.

A call to \Call{Check}{$F$} searches for the alternative derivations of $F$
using backward chaining. The function maintains the set $C$ of \emph{checked}
facts, which ensures that each $F$ is checked only once
(line~\ref{alg:bfcnt:check:FinBC} and~\ref{alg:bfcnt:saturate:C}). The
procedure first calls \Call{Saturate}{$F$} to determine whether $F$ follows
from the facts considered thus far; we discuss this step in more detail
shortly. If $F$ is not proved, the procedure then examines in
lines~\ref{alg:bfcnt:check:backward-chaining}--\ref{alg:bfcnt:check:backward-chaining:end}
each instance $r'$ of a recursive rule that derives $F$ in the `old'
materialisation, and it tries to prove all body atoms of $r'$ from the current
stratum. This involves evaluating rules `backwards' and, as we already
discussed in Section~\ref{sec:bfcnt}, this is the main advantage of the B/F
algorithm over DRed on a number of complex inputs. The function terminates once
$F$ is successfully proved (line~\ref{alg:bfcnt:check:backward-chaining:end}).

Set $P$ accumulates facts that are checked and successfully proved, and it is
computed in function \Call{Saturate}{} using forward chaining. Given a fact $F$
that is being checked, it first verifies whether $F$ has a nonrecursive
derivation. In the original B/F algorithm, this is done by evaluating the
nonrecursive rules `backwards' in the same way as in
line~\ref{alg:dred:OneStepRederive} of DRed. In contrast, B/F$^c$ avoids this
by simply checking whether the nonrecursive counter is nonzero
(line~\ref{alg:bfcnt:saturate:base}): if that is the case, then $F$ is known to
have nonrecursive derivations and it is added to $P$ via
lines~\ref{alg:bfcnt:saturate:np:init}
and~\ref{alg:bfcnt:saturate:loop:DeltaY}. If $F$ is proved, the procedure
propagates its consequences
(line~\ref{alg:bfcnt:saturate:np:init}--\ref{alg:bfcnt:saturate:loop:end}). In
particular, the procedure ensures that each consequence $F'$ of $P$, the facts
in the `new' materialisation in the previous strata, and the recursive rules is
added to $P$ if ${F' \in C}$, or is added to the set $Y$ of \emph{delayed}
facts if ${F' \not\in C}$. Intuitively, set $Y$ contains facts that are proved
but that have not been checked yet. If a fact in $Y$ is checked at a later
point, it is proved in line~\ref{alg:bfcnt:saturate:base} without having to
apply the rules again.

Since the deletion step of B/F$^c$ is `exact' in the sense that it deletes
precisely those facts that no longer hold after the update, rederivation is not
needed. Thus, \Call{DeleteUnproved}{} is directly followed by \Call{Insert}{},
which is the same as in DRed and DRed$^c$, with the only difference that
B/F$^c$ maintains only the nonrecursive counters.

\begin{algorithm}[!tb]
\fontsize{8}{10}\selectfont
\caption{\textsc{B/F$^c$}$(\Pi, \lambda, E, I, E^-, E^+, \Cnr)$}\label{alg:bfcnt}
\begin{algorithmiccont}
    \State $D \defeq A \defeq \emptyset$, \quad $E^- = (E^- \cap E) \setminus E^+$, \quad $E^+ = E^+ \setminus E$                                               \label{alg:bfcnt:init}
    \For{each stratum index $s$ with ${1 \leq s \leq S}$}                                                                                                       \label{alg:bfcnt:stratum:start}
        \State $C \defeq P \defeq Y \defeq \emptyset$
        \State \Call{DeleteUnproved}{}                                                                                                                          \label{alg:bfcnt:DeleteUnproved}
        \State \Call{Insert}{}                                                                                                                                  \label{alg:bfcnt:Insert}
    \EndFor                                                                                                                                                     \label{alg:bfcnt:stratum:end}
    \State $E \defeq (E \setminus E^-) \cup E^+$, \quad $I \defeq (I \setminus D) \cup A$                                                                       \label{alg:bfcnt:update-I}
    \Statex
    \Procedure{DeleteUnproved}{}
        \State $N_D \defeq \emptyset$
        \For{$F \in (E^- \cap \Out{s}) \cup \mapply{\nrstrat{\Pi}{s}}{I \appargs D \setminus A, A \setminus D}$}                                                \label{alg:bfcnt:del:decnrcntloop}
            \State $N_D \defeq N_D \cup \{ F \}$, \quad $\Cnr[F] \defeq \Cnr[F] - 1$                                                                            \label{alg:bfcnt:del:decnrcntloop:end}
        \EndFor
        \State $N_D \defeq N_D \cup \apply{\rstrat{\Pi}{s}}{I \appargs D \setminus A, A \setminus D}$
        \Loop                                                                                                                                                   \label{alg:bfcnt:del:loop}
            \State $\Delta_D \defeq \emptyset$                                                                                                                  \label{alg:bfcnt:del:deltaD-empty}
            \For{$F \in N_D \setminus D$}                                                                                                                       \label{alg:bfcnt:del:foreach}
                \State $\Call{Check}{F}$                                                                                                                        \label{alg:bfcnt:del:check-init}
                \If{$F \not\in P$}
                    $\Delta_D \defeq \Delta_D \cup \{ F \}$                                                                                                     \label{alg:bfcnt:del:deltaD}
                \EndIf
            \EndFor                                                                                                                                             \label{alg:bfcnt:del:foreach:end}
            \If{$\Delta_D = \emptyset$}
                \Break
            \EndIf
            \State $N_D \defeq \apply{\rstrat{\Pi}{s}}{I \setminus (D \setminus A), I \cup A \appargs \Delta_D}$                                                \label{alg:bfcnt:del:ND-update}
            \State $D \defeq D \cup \Delta_D$                                                                                                                   \label{alg:bfcnt:del:D-update}
        \EndLoop                                                                                                                                                \label{alg:bfcnt:del:endloop}
    \EndProcedure
    \Statex
    \Function{Check}{$F$}
        \If{$F \not\in C$}                                                                                                                                      \label{alg:bfcnt:check:FinBC}
            \If{$\Call{Saturate}{F} = \false$}                                                                                                                  \label{alg:bfcnt:check:saturate}
                \For{each $r \in \rstrat{\Pi}{s}$ and each \Statex \hspace{1cm} $r' \in \ins{r}{I \setminus ((D \cup \Delta_D) \setminus A), I \cup A}$ s.t.\ $\head{r'} = F$}   \label{alg:bfcnt:check:backward-chaining}
                    \For{$G \in \pbody{r'} \cap \Out{s}$}                                                                                                       \label{alg:bfcnt:check:for}
                        \State \Call{Check}{$G$}                                                                                                                \label{alg:bfcnt:check:check}
                        \If{$F \in P$}
                            \Return                                                                                                                             \label{alg:bfcnt:check:lastreturn}
                        \EndIf
                    \EndFor                                                                                                                                     \label{alg:bfcnt:check:for:end}
                \EndFor                                                                                                                                         \label{alg:bfcnt:check:backward-chaining:end}
            \EndIf
        \EndIf
    \EndFunction
    \Statex
    \Function{Saturate}{$F$}
        \State $C \defeq C \cup \{ F \}$                                                                                                                        \label{alg:bfcnt:saturate:C}
        \If{$F \in Y$ or $\Cnr[F] > 0$}                                                                                                                         \label{alg:bfcnt:saturate:base}
            \State $N_P \defeq \{ F \}$                                                                                                                         \label{alg:bfcnt:saturate:np:init}
            \Loop                                                                                                                                               \label{alg:bfcnt:saturate:loop}
                \State $\Delta_P \defeq (N_P \cap C) \setminus P$, \quad $Y \defeq Y \cup N_P \setminus C$                                                      \label{alg:bfcnt:saturate:loop:DeltaY}
                \If{$\Delta_P = \emptyset$}                                                                                                                     \label{alg:bfcnt:saturate:loop:if}
                    \Return $\true$
                \EndIf
                \State $P \defeq P \cup \Delta_P$                                                                                                               \label{alg:bfcnt:saturate:P}
                \State $N_P \defeq \apply{\rstrat{\Pi}{s}}{P \cup (\Out{<s} \cap (I \setminus (D \setminus A))), I \cup A \appargs \Delta_P}$                   \label{alg:bfcnt:saturate:np:update}
            \EndLoop                                                                                                                                            \label{alg:bfcnt:saturate:loop:end}
        \Else
            \ \Return $\false$
        \EndIf
    \EndFunction
    \Statex
    \Procedure{Insert}{}
        \State $N_A \defeq \emptyset$
        \For{$F \in (E^+ \cap \Out{s}) \oplus \mapply{\nrstrat{\Pi}{s}}{(I \setminus D) \cup A \appargs A \setminus D, D \setminus A}$}                         \label{alg:bfcnt:ins:incnrloop}
            \State $N_A \defeq N_A \cup \{ F \}$, \quad $\Cnr[F] \defeq \Cnr[F] + 1$                                                                            \label{alg:bfcnt:ins:incnrloop:end}
        \EndFor
        \State $N_A \defeq N_A \cup \apply{\nrstrat{\Pi}{s}}{(I \setminus D) \cup A \appargs A \setminus D, D \setminus A}$
        \Loop                                                                                                                                                   \label{alg:bfcnt:ins:loop}
            \State $\Delta_A \defeq N_A \setminus ((I \setminus D) \cup A)$                                                                                     \label{alg:bfcnt:ins:DeltaA}
            \If{$\Delta_A = \emptyset$}
                \Break
            \EndIf
            \State $A \defeq A \cup \Delta_A$                                                                                                                   \label{alg:bfcnt:ins:A-update}
            \State $N_A \defeq N_A \cup \apply{\rstrat{\Pi}{s}}{(I \setminus D) \cup A \appargs \Delta_A}$                                                      \label{alg:bfcnt:ins:NA-update}
        \EndLoop                                                                                                                                                \label{alg:bfcnt:ins:loop:end}
    \EndProcedure
\end{algorithmiccont}
\end{algorithm}

Algorithm~\ref{alg:bfcnt} is correct in the same way as B/F since checking 
whether a fact has a nonzero nonrecursive counter is equivalent to checking 
whether a derivation of the fact exists by evaluating nonrecursve rules 
`backwards'.

\section{Evaluation}\label{sec:evaluation}

We have implemented the unoptimised and optimised variants of DRed and B/F and
have compared them empirically.

\subsubsection{Benchmarks}

We used the following benchmarks for our evaluation: UOBM~\cite{ma2006towards}
is a synthetic benchmark that extends the well known LUBM~\cite{guo2005lubm}
benchmark; Reactome~\cite{croft2013reactome} models biological pathways of
molecules in cells; Uniprot~\cite{bateman2015uniprot} describes protein
sequences and their functional information; ChemBL~\cite{gaulton2011chembl}
represents functional and chemical properties of bioactive compounds; and
Claros describes archeological artefacts. Each benchmark consists of a set of
facts and an OWL 2 DL ontology, which we transformed into datalog programs of
different levels of complexity and recursiveness. More specifically, the
\emph{upper bound} (U) programs were obtained using the complete but unsound
transformation by \citet{zhou2013making}, and they entail all consequences of
the original ontology but may also derive additional facts. The
\emph{recursive} (R) programs were obtained using the sound but incomplete
transformation by \citet{kaminski2016datalog}, and they tend to be highly
recursive. For Claros, the \emph{lower bound extended} (LE) program was
obtained by manually introducing several `hard' rules, and it was already used
by \citet{motik2015incremental} to compare DRed with B/F. Finally, to estimate
the effect of built-in literals on materialisation maintenance, we developed a
new synthetic benchmark SSPE (Single-Source Path Enumeration). Its dataset
consists of a randomly generated directed acyclic graph of 100~k nodes and 1~M
edges, and its program traverses paths from a single source analogously to
rules \eqref{ex2rule1}--\eqref{ex2rule2}. All the tested programs are
recursive, although the percentage of the recursive rules varies.
Table~\ref{t1} shows the numbers of facts ($|E|$), strata ($S$), the
nonrecursive rules ($|\nrstrat{\Pi}{}|$), and the recursive ones
($|\rstrat{\Pi}{}|$) for each benchmark.

\begin{table*}
\centering
\begin{tabular}{cc|r|r|r|r|r|r|r|r}
    \multicolumn{2}{c|}{Dataset}            & $\abs{E}$ & $S$  & $|\nrstrat{\Pi}{}|$   & $|\rstrat{\Pi}{}|$    & DRed$^c$  & DRed      & B/F$^c$   & B/F \\
    \hline
    \multicolumn{2}{c|}{UOBM-U}             & 254.8 M   & 5    & 135                   & 144                   & 179.24    & 1,185.87  & 1.18      & 31.96 \\
    \hline
    \multicolumn{2}{c|}{UOBM-R}             & 254.8 M   & 6    & 164                   & 2,215                 & 0.29      & 0.56      & 0.26      & 0.24 \\
    \hline
    \multicolumn{2}{c|}{Reactome-U}         & 12.5 M    & 9    & 814                   & 28                    & 0.06      & 1.03      & 0.05      & 1.00 \\
    \hline
    \multicolumn{2}{c|}{Reactome-R}         & 12.5 M    & 1    & 0                     & 21,385                & 26.57     & 62.46     & 0.89      & 0.90 \\
    \hline
    \multicolumn{2}{c|}{Uniprot-R}          & 123.1 M   & 5    & 9,312                 & 2,706                 & 4.31      & 8.71      & 4.13      & 4.36 \\
    \hline
    \multicolumn{2}{c|}{ChemBL-R}           & 289.2 M   & 3    & 1,766                 & 499                   & 7.91      & 12.22     & 1.27      & 1.26 \\
    \hline
    \multirow{3}{*}{Claros-LE}  & Best      &           &      &                       &                       & 0.33      & 5,788.75  & 0.43      & 8.03 \\
                                & Worst     & 18.8 M    & 11   & 1,031                 & 306                   & 5,720.57  & 5,759.92  & 2,802.86  & 3,227.05 \\
                                & Average   &           &      &                       &                       & 1,143.55  & 1,741.66  & 560.61    & 653.04 \\
    \hline
    \multicolumn{2}{c|}{SSPE}                & 3.0 M     & 1    & 1                     & 1                     & 10.53     & 1,684.97  & 247.00    & 252.97 \\
\end{tabular}
\caption{Average running times for deleting 1000 facts (seconds)}\label{t1}
\end{table*}

\begin{table*}
  \centering
  \newcommand{\pk}{\phantom{k}}
  \begin{tabular}{c|r|r|r|r|r|r|r|r}
    Dataset     & $\abs{E^-}/\abs{E}$   & DRed$^c$  & DRed      & B/F$^c$   & B/F       & Remat     & Remat-1C  & Remat-2C \\
    \hline
    UOBM-U      & 50\%                  & 1.54 k    & 3.66 k    & 1.64 k    & 3.11 k    & 1.56 k    & 1.60 k    & 1.61 k \\
    \hline
    UOBM-R      & 36\%                  & 3.28 k    & 5.75 k    & 2.76 k    & 2.87 k    & 4.14 k    & 4.16 k    & 4.19 k \\
    \hline
    Reactome-U  & 68\%                  & 30.70 \pk & 6.32 k    & 39.16 \pk & 6.33 k    & 30.90 \pk & 31.23 \pk & 31.32 \pk \\
    \hline
    Reactome-R  & 31\%                  & 1.07 k    & 1.78 k    & 0.92 k    & 0.93 k    & 0.91 k    &  0.91 k   & 0.92 k \\
    \hline
    Uniprot-R   & 47\%                  & 1.57 k    & 3.47 k    & 1.92 k    & 2.86 k    & 1.98 k    & 1.99 k    & 2.00 k \\
    \hline
    ChemBL-R    & 69\%                  & 4.74 k    & 7.22 k    & 3.25 k    & 4.18 k    & 2.56 k    & 2.57 k    & 2.59 k \\
    \hline
    Claros-LE   & 8\%                   & 5.01 k    & 17.81 k   & 3.49 k    & 16.74 k   & 3.36 k    & 3.49 k    & 3.60 k \\
    \hline
    SSPE         & 2\%                   & 74.36 \pk & 7.24 k    & 7.90 k    & 7.75 k    & 68.28 \pk & 70.18 \pk & 71.81 \pk \\
  \end{tabular}
  \caption{Running times for handling large deletions (seconds)}\label{t2}
\end{table*}

\subsubsection{Test Setup}

We conducted all experiments on a Dell PowerEdge R720 server with 256GB RAM and
two Intel Xeon E5-2670 2.6GHz processors, running Fedora 24, kernel version
4.8.12-200.fc24.x86\_64. All algorithms handle insertions using the semina\"ive
evaluation. The only overhead is in counter maintenance, which we measured
during initial materialisation (which also uses semina\"ive evaluation). Hence,
the main focus of our tests was on comparing the performance of our algorithms
on `small' and `large' deletions. In both cases, we first materialised the
relevant program on the explicit facts, and then we performed the following
tests.

To test small deletions, we measured the performance on ten randomly selected
subsets ${E^- \subseteq E}$ of 1,000 facts. In all apart from Claros-LE, the
running times did not depend significantly on the selected subset of $E$, so in
Table~\ref{t1} we report the average times across all ten runs. On Claros-LE,
however, the running times varied significantly, so we report in the table the
best, the worst, and the average times.

To test large deletions, we identified the largest subset ${E^- \subseteq E}$
on which either DRed$^c$ or B/F$^c$ takes roughly the same time as computing
the `new' materialisation from scratch. We measured the performance of all
algorithms on $E^-$, as well as the performance of rematerialisation with no
counters (Remat), just the nonrecursive counter (Remat-1C), and both counters
(Remat-2C). This test allows us to assess the scalability of our algorithms.
Table~\ref{t2} reports the running times and the percentages of the deleted
facts.

\subsubsection{Discussion}

DRed$^c$ outperformed DRed on all inputs for small deletions. In particular, on
SSPE, the average running time for DRed drops from 28 minutes to just over 10
seconds. On Reactome-U, the improvement is by several orders of magnitude,
albeit unoptimised DRed is already quite efficient. The improvement is also
significant in many other cases, including UOBM-U, Reactome-R, and ChemBL-R. In
fact, Reactome-U and SSPE exhibit data and rule patterns outlined in
Examples~\ref{ex1} and~\ref{ex2}, which clearly demonstrates the benefits of
eliminating `backward' rule evaluation. Moreover, the program of Claros-LE
contains a symmetric and transitive predicate $\mathit{relatedPlaces}$, so the
materialisation contains several large cliques of constants connected by this
predicate \cite{motik2015incremental}. When a fact
$\mathit{relatedPlaces}(a,b)$ is overdeleted, the DRed algorithm overdeletes
all $\mathit{relatedPlaces}(c,d)$ where $c$ and $d$ belong to the same clique
as $a$ and $b$, which requires a cubic number of derivations. However, DRed$^c$
can sometimes (but not always) prove that $\mathit{relatedPlaces}(a,b)$ holds
using the nonrecursive counter; as one can see, this can considerably improve
the performance by avoiding costly overdeletion.

B/F$^c$ also outperformed B/F for small deletions in many cases: B/F$^c$ was
more than 20 times faster for UOBM-U, Reactome-U, and the `best' case of
Claros-LE, which is in line with our observation that `backwards' rule
evaluation can be quite costly. In contrast, on the highly recursive datasets
(i.e., all R-datasets and SSPE), the performance of B/F$^c$ and B/F is roughly
the same: the main source of difficulty is due to the recursive rules, whose
evaluation is unaffected by the optimisations proposed in this paper.

B/F$^c$ outperformed DRed$^c$ on all datasets but SSPE. This is so because
B/F$^c$ eagerly identifies alternative derivations of facts, which is often
easy, and is beneficial since it can considerably reduce overdeletion. However,
as we discussed earlier in Section~\ref{sec:motivation}, backward rule
evaluation can be a dominant source of inefficiencty (e.g., on SSPE). In such
cases, DRed$^c$ is faster than B/F$^c$ since DRed$^c$ completely eliminates
backward rule evaluation, whereas B/F$^c$ only avoids backward evaluation on
nonrecursive rules.

The tests for large deletions show that our algorithms can efficiently delete
large subsets of the explicit facts on all but two benchmarks: Claros-LE and
SSPE. Claros-LE is difficult due to the presence of cliques as explained
earlier, and SSPE is difficult because deleting a small percentage of the
explicit facts leads to the deletion of about half of the inferred facts.
Nevertheless, DRed$^c$ always considerably outperforms DRed; the difference is
particularly significant on Reactome-U and SSPE, where DRed$^c$ is several
orders of magnitude faster. Similarly, B/F$^c$ consistently outperforms B/F on
all cases apart from SSPE, where the latter is due to the overhead of
maintaining the counters.

Finally, the rematerialisation times show that counter maintenance incurs only
modest overheads: Remat-2C was in the worst case only several percent slower
than Remat.

\section{Conclusion}\label{sec:conclusion}

We have presented two novel algorithms for the maintenance of datalog
materialisations, obtained by combining the well-known DRed and B/F algorithms
with Counting. Our evaluation shows that our algorithms are generally more
efficient than the original ones, often by orders of magnitude. Our algorithms
could handle both small and large updates efficiently, and have thus been shown
to be ready for practical use. In future, we shall develop a modular approach
to materialisation and its maintenance that tackles the difficult cases such as
Claros-LE using reasoning modules that can be `plugged into' the semina\"ive
evaluation to handle difficult rule combinations using custom algorithms.

\section*{Acknowledgments}

This work was supported by the EPSRC projects MaSI$^3$, DBOnto, and ED$^3$.

\bibliographystyle{aaai}
\bibliography{references}

\iftoggle{withappendix}{
    \newpage
    \onecolumn
    \appendix
    \section{Proof of Correctness for Algorithm~\ref{alg:dredcnt}}\label{sec:proof}

\dredcntcorrectness*

\begin{proof}
Let $\ominus$ be the multiset subtraction operator, and let $\mathsf{Occ}(F, M)$ be the multiplicity of $F$ in multiset $M$. Due to line~\ref{alg:dredcnt:init} of Algorithm~\ref{alg:dredcnt}, without loss of generality we assume that ${E^- \subseteq E}$ and ${E^+ \cap E = \emptyset}$. Now let ${\Old{E} = E}$ and let ${\Old{I^0} = \emptyset}$. Moreover, for each $1 \leq s \leq S$, let $\Old{I^s_0}, \Old{I^s_1}, \dots$ be the sequence of sets where ${\Old{I^s_0} = \Old{I^{s-1}} \cup (\Old{E} \cap \Out{s})}$, and for $i > 0$, ${\Old{I^s_i} = \Old{I^s_{i-1}} \cup \apply{\strat{\Pi}{s}}{\Old{I^s_{i-1}}}}$. Index $k$ clearly exists at which the sequence reaches the fixpoint (i.e., $\Old{I^s_k} = \Old{I^s_{k+1}}$), so let $\Old{I^s} = \Old{I^s_k}$. Finally, let $\Old{I} = \Old{I^S}$; we clearly have $\Old{I} = \mat{\Pi}{\Old{E}}$---that is, $\Old{I}$ is the `old' materialisation. Now let $\New{E} = (\Old{E} \setminus E^-) \cup E^+$, and let $\New{I^s_i}$, $\New{I^s}$, and $\New{I}$ be defined analogously, so $\New{I}$ is the `new' materialisation.

For each $1 \leq s \leq S$ and each $F \in \Old{I} \cap \Out{s}$, let $\Old{\Cnr[F]} = \mathsf{Occ}(F, \Old{E} \oplus \mapply{\nrstrat{\Pi}{s}}{\Old{I}})$ and $\Old{\Cr[F]} = \mathsf{Occ}(F, \mapply{\rstrat{\Pi}{s}}{\Old{I}})$; we define $\New{\Cnr[F]}$ and $\New{\Cr[F]}$ analogously using $\New{I}$ and $\New{E}$.

Now consider a run of Algorithm~\ref{alg:dredcnt} on $\Old{I}$, $E^-$, and $E^+$. Let $D^0 = A^0 = R^0 = \emptyset$, and for each $s$ with $1 \leq s \leq S$, let $D^s$, $A^s$, and $R^s$ be the values of $D$, $A$, and $R$, respectively, after the loop in lines~\ref{alg:dredcnt:stratum-loop}--\ref{alg:dredcnt:stratum-loop:end} finishes for stratum index $s$. Note that during the execution of Algorithm~\ref{alg:dredcnt}, the set $I$ is equal to $\Old{I}$ up to before line~\ref{alg:dredcnt:update-I}. Furthermore, for each fact $F \in \Old{I^s} \cap \Out{s}$, let $\Del{\Cnr[F]}$ and $\Del{\Cr[F]}$ be the values of $\Cnr[F]$ and $\Cr[F]$, respectively, at the point when $\Call{Overdelete}{}$ finishes for stratum index $s$; similarly, for each fact $F \in ((\Old{I^s} \setminus D^s) \cup A^s) \cap \Out{s}$, let $\Add{\Cnr[F]}$ and $\Add{\Cr[F]}$ be the values of $\Cnr[F]$ and $\Cr[F]$, respectively, at the point when $\Call{Insert}{}$ finishes for stratum index $s$.

\begin{align}
\begin{array}{@{}l@{\;}l@{\;}l@{}}
\mapply{\nrstrat{\Pi}{s}}{\Old{I^s}} = \mapply{\nrstrat{\Pi}{s}}{\Old{I^s} \appargs D^s \setminus A^{s-1}, A^{s-1} \setminus D^s} \oplus \mapply{\nrstrat{\Pi}{s}}{\Old{I^s} \setminus (D^s \setminus A^{s-1}), \Old{I^s} \cup A^{s-1}}
\end{array} & \label{nonrecursiveoldequals} \\
\begin{array}{@{}l@{\;}l@{\;}l@{}}
\mapply{\rstrat{\Pi}{s}}{\Old{I^s}} = \mapply{\rstrat{\Pi}{s}}{\Old{I^s} \appargs D^s \setminus A^{s-1}, A^{s-1} \setminus D^s} \oplus \mapply{\rstrat{\Pi}{s}}{\Old{I^s} \setminus (D^s \setminus A^{s-1}), \Old{I^s} \cup A^{s-1}}
\end{array} & \label{recursiveoldequals} \\
\begin{array}{@{}r@{\;}r@{\;}r@{}}
\Old{\Cnr[F]}-\Del{\Cnr[F]} = \mathsf{Occ}(F, (E^- \cap \Out{s}) \oplus \mapply{\nrstrat{\Pi}{s}}{\Old{I^s} \appargs D^s \setminus A^{s-1}, A^{s-1} \setminus D^s}) \\
\text{for each } F \in \Old{I^s} \cap \Out{s}
\end{array} & \label{nonrecursivecntdecrement}\\
\begin{array}{@{}r@{\;}r@{\;}r@{}}
\Old{\Cr[F]}-\Del{\Cr[F]} = \mathsf{Occ}(F, \mapply{\rstrat{\Pi}{s}}{\Old{I^s} \appargs D^s \setminus A^{s-1}, A^{s-1} \setminus D^s}) \\
\text{for each } F \in \Old{I^s} \cap \Out{s}
\end{array} & \label{recursivecntdecrement}\\
\begin{array}{@{}l@{\;}l@{\;}l@{}}
\Old{I^s} \setminus \New{I^s} \subseteq D^s \subseteq \Old{I^s}
\end{array} & \label{deletionproperty} \\
\begin{array}{@{}l@{\;}l@{\;}l@{}}
\Out{s} \cap D^s \cap \apply{\rstrat{\Pi}{s}}{\Old{I^s} \setminus (D^s \setminus A^{s-1}), \Old{I^s} \cup A^{s-1}} \subseteq R^s \subseteq \New{I^s}
\end{array} & \label{onesteprederivationproperty} \\
\begin{array}{@{}l@{\;}l@{\;}l@{}}
\Old{I^s} \cap A^s \subseteq D^s
\end{array} & \label{ioldcapaind} \\
\begin{array}{@{}l@{\;}l@{\;}l@{}}
(\Old{I^s} \setminus D^s) \cup A^s = \New{I^s}
\end{array} & \label{insertionproperty} \\
\begin{array}{@{}r@{\;}r@{\;}r@{}}
\Add{\Cnr[F]}-\Del{\Cnr[F]} = \mathsf{Occ}(F, (E^+ \cap \Out{s}) \oplus \mapply{\nrstrat{\Pi}{s}}{(\Old{I^s} \setminus D^s) \cup A^s} \ominus \mapply{\nrstrat{\Pi}{s}}{\Old{I^s} \setminus (D^s \setminus A^{s-1}), \Old{I^s} \cup A^{s-1}}) \\
\text{for each } F \in ((\Old{I^s} \setminus D^s) \cup A^s) \cap \Old{I^s} \cap \Out{s}
\end{array} & \label{nonrecursivecntincrement}\\
\begin{array}{@{}r@{\;}r@{\;}r@{}}
\Add{\Cnr[F]} = \mathsf{Occ}(F, ((\Old{E} \setminus E^-) \cup E^+) \oplus \mapply{\nrstrat{\Pi}{s}}{(\Old{I^s} \setminus D^s) \cup A^s}) \\
\text{for each } F \in (((\Old{I^s} \setminus D^s) \cup A^s) \setminus \Old{I^s}) \cap \Out{s}
\end{array} & \label{nonrecursivenew}\\
\begin{array}{@{}r@{\;}r@{\;}r@{}}
\Add{\Cr[F]}-\Del{\Cr[F]} = \mathsf{Occ}(F, \mapply{\rstrat{\Pi}{s}}{(\Old{I^s} \setminus D^s) \cup A^s} \ominus \mapply{\rstrat{\Pi}{s}}{\Old{I^s} \setminus (D^s \setminus A^{s-1}), \Old{I^s} \cup A^{s-1}}) \\
\text{for each } F \in ((\Old{I^s} \setminus D^s) \cup A^s) \cap \Old{I^s} \cap \Out{s}
\end{array} & \label{recursivecntincrement}\\
\begin{array}{@{}r@{\;}r@{\;}r@{}}
\Add{\Cr[F]} = \mathsf{Occ}(F, \mapply{\rstrat{\Pi}{s}}{(\Old{I^s} \setminus D^s) \cup A^s}) \\
\text{for each } F \in (((\Old{I^s} \setminus D^s) \cup A^s) \setminus \Old{I^s}) \cap \Out{s}
\end{array} & \label{recursivenew}
\end{align}

We shall prove that properties~(\ref{deletionproperty}), (\ref{ioldcapaind}), and (\ref{insertionproperty}) hold for each $s$ with $0 \leq s \leq S$, and that the other properties hold for each $s$ with $1 \leq s \leq S$. Then, property~(\ref{insertionproperty}) for $s=S$ and line~\ref{alg:dredcnt:update-I} of Algorithm~\ref{alg:dredcnt} imply that $I$ is correctly updated from $\Old{I}$ to $\New{I}$. Moreover, ${\Old{\Cnr[F]} = \mathsf{Occ}(F, \Old{E} \oplus \mapply{\nrstrat{\Pi}{s}}{\Old{I}})}$ and properties~(\ref{insertionproperty}), (\ref{nonrecursiveoldequals}), (\ref{nonrecursivecntdecrement}) and (\ref{nonrecursivecntincrement}) for $1 \leq s \leq S$ jointly imply the following:
\begin{align}
\begin{array}{@{}r@{\;}r@{\;}r@{}}
\Add{\Cnr[F]} = \Old{\Cnr[F]} - (\Old{\Cnr[F]} - \Del{\Cnr[F]}) + (\Add{\Cnr[F]} - \Del{\Cnr[F]})  = \mathsf{Occ}(F, \New{E} \oplus \mapply{\nrstrat{\Pi}{s}}{\New{I}}) = \New{\Cnr[F]} \\
\text{for each } F \in \New{I^s} \cap \Old{I^s} \cap \Out{s}
\end{array}
\end{align}
In addition, (\ref{insertionproperty}) and (\ref{nonrecursivenew}) imply $\Add{\Cnr[F]} = \New{\Cnr[F]}$ for each $F \in (\New{I^s} \setminus \Old{I^s}) \cap \Out{s}$. Therefore, we have $\Add{\Cnr[F]} = \New{\Cnr[F]}$ for each $F \in \New{I^s} \cap \Out{s}$, which means the nonrecursive counts are correctly updated after the execution of the algorithm. $\Add{\Cr[F]} = \New{\Cr[F]}$ can be shown analogously using properties~(\ref{insertionproperty}), (\ref{recursiveoldequals}), (\ref{recursivecntdecrement}), (\ref{recursivecntincrement}), and (\ref{recursivenew}).

We prove properties~(\ref{nonrecursiveoldequals})--(\ref{recursivenew}) by induction on $s$. The base case where $s=0$ is trivial since all relevant sets in (\ref{deletionproperty}), (\ref{ioldcapaind}), and (\ref{insertionproperty}) are empty. For the inductive step, we consider an arbitrary $s$ with $1 \leq s \leq S$ such that (\ref{deletionproperty}), (\ref{ioldcapaind}), and (\ref{insertionproperty}) hold for $s-1$, and we show that properties (\ref{nonrecursiveoldequals})--(\ref{recursivenew}) hold for $s$. The proof is lengthy so we break it into several claims.
\end{proof}

\begin{claim} \label{claim-oldequals}
Properties~(\ref{nonrecursiveoldequals}) and (\ref{recursiveoldequals}) hold.
\end{claim}
\begin{proof}
The way sets $D$ and $A$ are constructed ensures that $(D^s \setminus D^{s-1}) \cap \Out{<s} = \emptyset$ and $A^{s-1} \subseteq \Out{<s}$ hold, and so ${A^{s-1} \setminus D^s = A^{s-1} \setminus D^{s-1}}$ holds as well, which in turn implies $\Old{I^s} \cup (A^{s-1} \setminus D^s) = \Old{I^s} \cup (A^{s-1} \setminus D^{s-1})$. Moreover, the induction assumption $D^{s-1} \subseteq \Old{I^{s-1}}$ ensures $\Old{I^s} \cup (A^{s-1} \setminus D^{s-1}) = \Old{I^s} \cup A^{s-1}$. Therefore, we have $\Old{I^s} \cup (A^{s-1} \setminus D^s) = \Old{I^s} \cup A^{s-1}$, which together with the definition of rule application ensures the correctness of the two properties.
\end{proof}

\begin{claim}
Property~(\ref{nonrecursivecntdecrement}) holds.
\end{claim}
\begin{proof}
Consider an arbitrary $F \in \Old{I^s} \cap \Out{s}$, we have two cases here.

If $F \not\in (E^- \cap \Out{s}) \oplus \mapply{\nrstrat{\Pi}{s}}{\Old{I^s} \appargs D^s \setminus A^{s-1}, A^{s-1} \setminus D^s}$, then the right-hand side of the equation in (\ref{nonrecursivecntdecrement}) equals zero. Moreover, $\nrstrat{\Pi}{s}$ contains only nonrecursive rules, and $A^{s-1} \setminus D^s = A^{s-1} \setminus D^{s-1}$ holds for the same reason as explained in the proof of claim~\ref{claim-oldequals}; thus we have $F \not\in (E^- \cap \Out{s}) \oplus \mapply{\nrstrat{\Pi}{s}}{\Old{I^s} \appargs D^{s-1} \setminus A^{s-1}, A^{s-1} \setminus D^{s-1}}$; but then, line~\ref{alg:dredcnt:del:decnrloop} of Algorithm~\ref{alg:dredcnt} ensures that the nonrecursive count of $F$ is not decremented during the execution of the $\Call{Overdelete}{}$ procedure, so the left-hand side of the equation is equal to zero as well. Therefore the property holds in this case.

If $F \in (E^- \cap \Out{s}) \oplus \mapply{\nrstrat{\Pi}{s}}{\Old{I^s} \appargs D^s \setminus A^{s-1}, A^{s-1} \setminus D^s} = (E^- \cap \Out{s}) \oplus \mapply{\nrstrat{\Pi}{s}}{\Old{I^s} \appargs D^{s-1} \setminus A^{s-1}, A^{s-1} \setminus D^{s-1}}$, then line~\ref{alg:dredcnt:del:decnrloop} and line~\ref{alg:dredcnt:del:decnrloop:end} guarantee that each distinct occurrence of $F$ in $(E^- \cap \Out{s}) \oplus \mapply{\nrstrat{\Pi}{s}}{\Old{I^s} \appargs D^{s-1} \setminus A^{s-1}, A^{s-1} \setminus D^{s-1}}$ results in decrementing the nonrecursive count of $F$ by one. Thus the property also holds in this case.
\end{proof}

\begin{claim}
Property~(\ref{recursivecntdecrement}) holds.
\end{claim}
\begin{proof}
Line~\ref{alg:dredcnt:del:DeltaD} and line~\ref{alg:dredcnt:del:D-update} ensure that $\Delta_D$ used in line~\ref{alg:dredcnt:del:ND-update} is different between two iterations of the loop in lines~\ref{alg:dredcnt:del:loop}-\ref{alg:dredcnt:del:endloop}, so the rule instances considered in line~\ref{alg:dredcnt:del:ND-update} are different between iterations. The total number of these rule instances is finite and is bounded by the number of rule instances in $\bigcup_{r' \in \rstrat{\Pi}{s}}\ins{r'}{\Old{I^s}}$. Thus the loop must terminate, and we let $T$ be the total number of iterations. Moreover, for each $1 \leq i \leq T$, let $D^s_i$ be the value of $D$ at the beginning of the $i$th iteration of the loop, and let $\Del{\Cr[F]}^i$ be the value of $\Cr[F]$ for each $F \in \Old{I^s} \cap \Out{s}$ at the same time point. We prove by induction on $i$ that (\ref{recursivecntdecrementinduction}) holds for $1 \leq i \leq T$. Then (\ref{recursivecntdecrementinduction}) for $i=T$ and $A^{s-1} \setminus D^{s-1} = A^{s-1} \setminus D^s$ ensure the correctness of property~(\ref{recursivecntdecrement}).
\begin{align}
\Old{\Cr[F]}-\Del{\Cr[F]}^i = \mathsf{Occ}(F, \mapply{\rstrat{\Pi}{s}}{\Old{I^s} \appargs D^s_i \setminus A^{s-1}, A^{s-1} \setminus D^{s-1}}) \text{ for each } F \in \Old{I^s} \cap \Out{s} \label{recursivecntdecrementinduction}
\end{align}

For the base case, consider an arbitrary $F \in \Old{I^s} \cap \Out{s}$. It is easy to see that the recursive count of $F$ has never been changed and should be equal to $\Old{\Cr[F]}$ before line~\ref{alg:dredcnt:del:decrloop} of procedure \Call{Overdelete}{} for stratum $s$. But then, line~\ref{alg:dredcnt:del:decrloop} and line~\ref{alg:dredcnt:del:decrloop:end} ensure that if $F \in \mapply{\rstrat{\Pi}{s}}{\Old{I^s} \appargs D^{s-1} \setminus A^{s-1}, A^{s-1} \setminus D^{s-1}}$, then each distinct occurrence of $F$ in the multiset results in decrementing the corresponding recursive count by one; moreover, if $F \not\in \mapply{\rstrat{\Pi}{s}}{\Old{I^s} \appargs D^{s-1} \setminus A^{s-1}, A^{s-1} \setminus D^{s-1}}$, then the corresponding recursive count will not be changed; either way, $D^s_1 = D^{s-1}$ implies $\Old{\Cr[F]}-\Del{\Cr[F]}^1 = \mathsf{Occ}(F, \mapply{\rstrat{\Pi}{s}}{\Old{I^s} \appargs D^s_1 \setminus A^{s-1}, A^{s-1} \setminus D^{s-1}})$, as required.

For the inductive step, assume that (\ref{recursivecntdecrementinduction}) holds for $i-1$ where $1 < i \leq T$, and consider arbitrary $F \in \Old{I^s} \cap \Out{s}$. Lines~\ref{alg:dredcnt:del:DeltaD}, \ref{alg:dredcnt:del:ND-update} and \ref{alg:dredcnt:del:D-update} jointly imply (\ref{deltadinduction}).
\begin{align}
\Del{\Cr[F]}^{i-1} - \Del{\Cr[F]}^{i} = \mathsf{Occ}(F, \mapply{\rstrat{\Pi}{s}}{\Old{I^s} \setminus (D^s_{i-1} \setminus A^{s-1}), \Old{I^s} \cup A^{s-1} \appargs D^s_{i} \setminus D^s_{i-1}}) \label{deltadinduction}
\end{align}
We now show that the following holds:
\begin{align}
\begin{array}{@{}r@{\;}r@{\;}r@{}}
\mathsf{Occ}(F, \mapply{\rstrat{\Pi}{s}}{\Old{I^s} \appargs D^s_{i-1} \setminus A^{s-1}, A^{s-1} \setminus D^{s-1}}) + \mathsf{Occ}(F, \mapply{\rstrat{\Pi}{s}}{\Old{I^s} \setminus (D^s_{i-1} \setminus A^{s-1}), \Old{I^s} \cup A^{s-1} \appargs D^s_{i} \setminus D^s_{i-1}}) \\
= \mathsf{Occ}(F, \mapply{\rstrat{\Pi}{s}}{\Old{I^s} \appargs D^s_i \setminus A^{s-1}, A^{s-1} \setminus D^{s-1}}) \label{cntdecrementinduction}
\end{array}
\end{align}
Between $\bigcup_{r' \in \rstrat{\Pi}{s}}\ins{r'}{\Old{I^s} \appargs D^s_{i-1} \setminus A^{s-1}, A^{s-1} \setminus D^{s-1}}$ and $\bigcup_{r' \in \rstrat{\Pi}{s}}\ins{r'}{\Old{I^s} \setminus (D^s_{i-1} \setminus A^{s-1}), \Old{I^s} \cup A^{s-1} \appargs D^s_{i} \setminus D^s_{i-1}}$ there is no rule instance repetition, so it is sufficient to show that (\ref{instanceequivalence}) holds.
\begin{align}
\begin{array}{@{}r@{\;}r@{\;}r@{}}
\bigcup_{r' \in \rstrat{\Pi}{s}}\ins{r'}{\Old{I^s} \appargs D^s_{i-1} \setminus A^{s-1}, A^{s-1} \setminus D^{s-1}} \cup \ins{r'}{\Old{I^s} \setminus (D^s_{i-1} \setminus A^{s-1}), \Old{I^s} \cup A^{s-1} \appargs D^s_{i} \setminus D^s_{i-1}} \\
= \bigcup_{r' \in \rstrat{\Pi}{s}}\ins{r'}{\Old{I^s} \appargs D^s_i \setminus A^{s-1}, A^{s-1} \setminus D^{s-1}}\label{instanceequivalence}
\end{array}
\end{align}
The $\subseteq$ direction of (\ref{instanceequivalence}) trivially holds. Now consider the $\supseteq$ direction, let $r''$ be an arbitrary rule instance contained in the right-hand side of (\ref{instanceequivalence}), then there exists rule $r' \in \rstrat{\Pi}{s}$ such that $r'' \in \ins{r'}{\Old{I^s} \appargs D^s_i \setminus A^{s-1}, A^{s-1} \setminus D^{s-1}}$ holds. If we have ${r'' \in \ins{r'}{\Old{I^s} \appargs D^s_{i-1} \setminus A^{s-1}, A^{s-1} \setminus D^{s-1}}}$, then clearly $r''$ is also contained in the left-hand side of (\ref{instanceequivalence}). Otherwise, $r''$ has all positive atoms in $\Old{I^s}$ but no positive atom in $D^s_{i-1} \setminus A^{s-1}$; moreover, all negative body atoms of $r''$ have to be matched in $\Old{I^s} \cup (A^{s-1} \setminus D^{s-1}) = \Old{I^s} \cup A^{s-1}$; furthermore, $r''$ has at least one positive body atom in $D^s_{i} \setminus D^s_{i-1}$; therefore, we have $r'' \in \ins{r'}{\Old{I^s} \setminus (D^s_{i-1} \setminus A^{s-1}), \Old{I^s} \cup A^{s-1} \appargs D^s_{i} \setminus D^s_{i-1}}$, so $r''$ is contained in the left-hand side of (\ref{instanceequivalence}) in this case as well.

(\ref{deltadinduction}), (\ref{cntdecrementinduction}) and the induction assumption that (\ref{recursivecntdecrementinduction}) holds for $i-1$ imply the correctness of (\ref{recursivecntdecrementinduction}) for $i$, and this completes our proof.
\end{proof}

\begin{claim} \label{claimrighthanddeletionproperty}
The right-hand inclusion of property~(\ref{deletionproperty}) holds.
\end{claim}
\begin{proof}
For each rule $r \in \strat{\Pi}{s}$, procedure \Call{Overdelete}{} for stratum index $s$ considers in lines~\ref{alg:dredcnt:del:decnrloop}, \ref{alg:dredcnt:del:decrloop}, and \ref{alg:dredcnt:del:ND-update} only instances of $r$ that are contained in $\ins{r}{\Old{I^s}}$, so the facts derived by these rule instances are in $\Old{I^s}$. Thus, the claim holds by a straightforward induction on the construction of the set $D^s$.
\end{proof}

\begin{claim}
The left-hand inclusion of property~(\ref{deletionproperty}) holds.
\end{claim}
\begin{proof}
We show by induction that \eqref{eq:fbf:proof:DcontainsDel:i} holds for each
$i$.
\begin{align}
    \traceOld[I]{s}{i} \setminus \traceNew[I]{s}{} \subseteq \trace[D]{s}{} \label{eq:fbf:proof:DcontainsDel:i}
\end{align}

For the base case, note that ${\traceOld[I]{s}{0} = \traceOld[I]{s-1}{} \cup (\Old{E} \cap \Out{s})}$ and that ${\traceOld[I]{s-1}{} \setminus \traceNew[I]{s-1}{} \subseteq \trace[D]{s-1}{}}$ holds for $s-1$ by the induction assumption. Now consider an arbitrary fact ${F \in \Old{E} \cap \Out{s}}$ such that ${F \not\in \traceNew[I]{s}{}}$ holds. Then, the latter ensures ${F \not\in \Old{E} \setminus E^-}$, which implies ${F \in E^-}$ and so $F$ is added to $N_D$ in line~\ref{alg:dredcnt:del:decnrloop:end}. We next show that $F$ is added to $D$ in line~\ref{alg:dredcnt:del:D-update}: ${F \not\in \traceNew[I]{s}{}}$ and ${\traceNew[I]{s}{} = \traceNew[I]{s}{} \cup \apply{\strat{\Pi}{s}}{\traceNew[I]{s}{}}}$ imply ${\mathsf{Occ}(F, \mapply{\nrstrat{\Pi}{s}}{\traceNew[I]{s}{}}) = 0}$, which in turn implies ${\mathsf{Occ}(F, \mapply{\nrstrat{\Pi}{s}}{\traceNew[I]{s-1}{}}) = 0}$; but then, ${\traceOld[I]{s-1}{} \setminus (D^{s-1} \setminus A^{s-1}) \subseteq (\traceOld[I]{s-1}{} \setminus D^{s-1}) \cup A^{s-1}}$ and the induction assumption ${(\traceOld[I]{s-1}{} \setminus D^{s-1}) \cup A^{s-1} = \traceNew[I]{s-1}{}}$ ensure ${\mathsf{Occ}(F, \mapply{\nrstrat{\Pi}{s}}{\traceOld[I]{s-1}{} \setminus (D^{s-1} \setminus A^{s-1})), \traceOld[I]{s-1}{} \cup A^{s-1}} = 0}$; together with (\ref{nonrecursiveoldequals}), (\ref{nonrecursivecntdecrement}), and the definition of $\Old{\Cnr[F]}$ this implies $\Del{\Cnr[F]} = 0$, so the condition in line~\ref{alg:dredcnt:del:DeltaD} is satisfied. Hence, ${F \in \trace[D]{s}{}}$ holds, as required.

For the inductive step, assume that $\traceOld[I]{s}{i-1}$ satisfies \eqref{eq:fbf:proof:DcontainsDel:i} for ${i > 0}$, and consider arbitrary ${F \in \traceOld[I]{s}{i} \setminus \traceNew[I]{s}{}}$. If ${F \in \traceOld[I]{s}{i-1}}$, then ${F \in \trace[D]{s}{}}$ holds by the induction assumption. Otherwise, there exist a rule ${r \in \strat{\Pi}{s}}$ and its instance ${r' \in \ins{r}{\traceOld[I]{s}{i-1}}}$ such that ${F \in \head{r'}}$. Definition \eqref{eq:ruleinstancepn} ensures ${\pbody{r'} \subseteq \traceOld[I]{s}{i-1} \subseteq \traceOld[I]{s}{}}$ and ${\nbody{r'} \cap \traceOld[I]{s}{i-1} = \emptyset}$, and ${\nbody{r'} \subseteq \Out{<s}}$ implies ${\nbody{r'} \cap \traceOld[I]{s-1}{} = \nbody{r'} \cap \traceOld[I]{s}{} = \emptyset}$. Finally, ${F \not\in \traceNew[I]{s}{}}$ implies ${r' \not\in \ins{r}{\traceNew[I]{s}{}}}$, so by definition \eqref{eq:ruleinstancepn} we have one of the following two possibilities.
\begin{itemize}
    \item ${\pbody{r'} \not\subseteq \traceNew[I]{s}{}}$. Thus, a fact ${G \in \pbody{r'}}$ exists such that ${G \in \traceOld[I]{s}{i-1} \setminus \traceNew[I]{s}{}}$ holds. The induction assumption for \eqref{eq:fbf:proof:DcontainsDel:i} implies ${G \in \trace[D]{s}{}}$, and ${G \not\in \traceNew[I]{s}{}}$ implies ${G \not\in \traceNew[I]{s-1}{}}$, so the induction assumption for \eqref{insertionproperty} ensures ${G \not\in \trace[A]{s-1}{}}$; hence, ${G \in \trace[D]{s}{} \setminus \trace[A]{s-1}{}}$.

    \item ${\nbody{r'} \cap \traceNew[I]{s}{} = \nbody{r'} \cap \traceNew[I]{s-1}{} \neq \emptyset}$. Thus, a fact ${G \in \nbody{r'}}$ exists such that ${G \in \traceNew[I]{s-1}{} \setminus \traceOld[I]{s-1}{}}$ holds; but then, the right-hand inclusion of \eqref{deletionproperty} implies ${G \not\in \trace[D]{s-1}{}}$, and the induction assumption for \eqref{insertionproperty} implies ${G \in \trace[A]{s-1}{}}$; therefore, we have ${G \in \trace[A]{s-1}{} \setminus \trace[D]{s-1}{} = \trace[A]{s-1}{} \setminus \trace[D]{s}{}}$.
\end{itemize}
Either way, the right-hand side of \eqref{nonrecursivecntdecrement} or \eqref{recursivecntdecrement} is larger than zero. Hence, the count for $F$ is decremented in $\Call{Overdelete}{}$ and $F$ is added to $N_D$. Then, in the same way as in the proof of the base case we have ${F \in \trace[D]{s}{}}$, as required.
\end{proof}

\begin{claim} \label{RcontainedinNewI}
The right-hand inclusion of property~\eqref{onesteprederivationproperty} holds.
\end{claim}
\begin{proof}
Consider an arbitrary $F \in R^s$. $F$ can only be added to $R^s$ in line~\ref{alg:dredcnt:OneStepRederive}, so we have $F \in D^s \cap \Out{s}$ and $\Del{\Cr[F]} > 0$; $F \in D^s \cap \Out{s}$ and property~\eqref{deletionproperty} ensure $F \in \Old{I^s} \cap \Out{s}$; but then, properties~\eqref{recursivecntdecrement}, \eqref{recursiveoldequals} and $\Old{\Cr[F]} = \mathsf{Occ}(F, \mapply{\rstrat{\Pi}{s}}{\Old{I^s}})$ jointly imply $\Del{\Cr[F]} = \mathsf{Occ}(F, \mapply{\rstrat{\Pi}{s}}{\Old{I^s} \setminus (D^s \setminus A^{s-1}), \Old{I^s} \cup A^{s-1}})$; now $\Del{\Cr[F]} > 0$ implies $\mathsf{Occ}(F, \mapply{\rstrat{\Pi}{s}}{\Old{I^s} \setminus (D^s \setminus A^{s-1}), \Old{I^s} \cup A^{s-1}}) > 0$. Next we show that $\mapply{\rstrat{\Pi}{s}}{\Old{I^s} \setminus (D^s \setminus A^{s-1}), \Old{I^s} \cup A^{s-1}} \subseteq \mapply{\rstrat{\Pi}{s}}{\New{I^s}}$: due to stratification negative body atoms will only be matched against facts from lower strata, so we have $\mapply{\rstrat{\Pi}{s}}{\New{I^s}} = \mapply{\rstrat{\Pi}{s}}{\New{I^s}, \New{I^{s-1}}}$; moreover, the left-hand inclusion of property~\eqref{deletionproperty} implies $\Old{I^s} \setminus D^s \subseteq \New{I^s}$, which together with the induction assumption for property~\eqref{insertionproperty} implies $\Old{I^s} \setminus (D^s \setminus A^{s-1}) \subseteq (\Old{I^s} \setminus D^s) \cup A^{s-1} \subseteq \New{I^s}$; furthermore, the induction assumption for property~\eqref{insertionproperty} ensures $\New{I^{s-1}} = (\Old{I^{s-1}} \setminus D^{s-1}) \cup A^{s-1} \subseteq \Old{I^s} \cup A^{s-1}$; therefore we clearly have $\mapply{\rstrat{\Pi}{s}}{\Old{I^s} \setminus (D^s \setminus A^{s-1}), \Old{I^s} \cup A^{s-1}} \subseteq \mapply{\rstrat{\Pi}{s}}{\New{I^s}}$. But then, $\mathsf{Occ}(F, \mapply{\rstrat{\Pi}{s}}{\Old{I^s} \setminus (D^s \setminus A^{s-1}), \Old{I^s} \cup A^{s-1}}) > 0$ implies $\mathsf{Occ}(F, \mapply{\rstrat{\Pi}{s}}{\New{I^s}}) > 0$, so we have $F \in \New{I^s}$, as required.
\end{proof}

\begin{claim}
The left-hand inclusion of property~\eqref{onesteprederivationproperty} holds.
\end{claim}
\begin{proof}
Consider arbitrary $F \in \Out{s} \cap D^s \cap \apply{\rstrat{\Pi}{s}}{\Old{I^s} \setminus (D^s \setminus A^{s-1}), \Old{I^s} \cup A^{s-1}}$, and we show that $F \in R^s$ holds. ${F \in \Out{s} \cap D^s}$ implies $\Del{\Cr[F]} = \mathsf{Occ}(F, \mapply{\rstrat{\Pi}{s}}{\Old{I^s} \setminus (D^s \setminus A^{s-1}), \Old{I^s} \cup A^{s-1}})$ in the same way as in the proof of claim~\ref{RcontainedinNewI}; but then, the fact that ${F \in \apply{\rstrat{\Pi}{s}}{\Old{I^s} \setminus (D^s \setminus A^{s-1}), \Old{I^s} \cup A^{s-1}}}$ holds imply $\Del{\Cr[F]} > 0$, so line~\ref{alg:dredcnt:OneStepRederive} of Algorithm~\ref{alg:dredcnt} ensures that $F$ is added to $R^s$, as required.
\end{proof}

\begin{claim} \label{claimforioldcapaind}
Property~\eqref{ioldcapaind} holds.
\end{claim}
\begin{proof}
Consider an arbitrary fact $F \in \Old{I^s} \cap A^s$: if $F \in \Old{I^s} \cap A^{s-1}$ holds, then the induction assumption ensures that we have ${F \in D^{s-1} \subseteq D^s}$; if $F \in \Old{I^s} \cap (A^s \setminus A^{s-1})$ holds, then $F$ is added to $\Delta_A$ in line~\ref{alg:dredcnt:ins:DeltaA} of procedure~\Call{Insert}{} for stratum index $s$, which ensures $F \not\in (\Old{I^s} \setminus D^s) \cup A^{s-1}$; but then, $F \in \Old{I^s}$ implies $F \in D^s$, as required.
\end{proof}

\begin{claim}
The $\subseteq$ direction of property~\eqref{insertionproperty} holds.
\end{claim}
\begin{proof}
We prove by induction on the construction of $A$ in $\textsc{Insert}$ that ${(\traceOld[I]{s}{} \setminus \trace[D]{s}{}) \cup A \subseteq \traceNew[I]{s}{}}$ holds. We first consider the base case. Set $A$ is equal to ${\trace[A]{s-1}{}}$ before the loop in lines~\ref{alg:dredcnt:ins:loop}--\ref{alg:dredcnt:ins:loop:end}; thus, property \eqref{insertionproperty} is equivalent to ${(\traceOld[I]{s}{} \setminus \trace[D]{s}{}) \cup \trace[A]{s-1}{} \subseteq \traceNew[I]{s}{}}$; now ${\traceOld[I]{s}{} \setminus \trace[D]{s}{} \subseteq \traceNew[I]{s}{}}$ is implied by the left-hand inclusion of property~\eqref{deletionproperty}, whereas $A^{s-1} \subseteq \New{I^s}$ holds by the induction assumption.

For the inductive step, we assume that $(\Old{I^s} \setminus D^s) \cup A \subseteq \New{I^s}$ holds, and we consider ways in which Algorithm~\ref{alg:dredcnt} can add a fact $F$ to $A$. If $F \in E^+ \cap \Out{s}$, then $F \in \New{I^s}$ clearly holds. Moreover, if $F \in R^s$ holds, then $F \in \New{I^s}$ holds by \eqref{onesteprederivationproperty}. Otherwise, $F$ is derived in line~\ref{alg:dredcnt:ins:incnrloop}, \ref{alg:dredcnt:ins:incrloop}, or \ref{alg:dredcnt:ins:incrsecondloop}, so a rule $r \in \strat{\Pi}{s}$ and its instance $r' = \ins{r}{(\Old{I^s} \setminus D^s) \cup A}$ exist such that $F \in \head{r'}$ holds. But then, definition~\eqref{eq:ruleinstancepn} ensures $\pbody{r'} \subseteq (\Old{I^s} \setminus D^s) \cup A \subseteq \New{I^s}$ and $\nbody{r'} \cap ((\Old{I^s} \setminus D^s) \cup A) = \emptyset$, which together with $\nbody{r'} \subseteq \Out{<s}$ and the induction assumption for \eqref{insertionproperty} implies $\nbody{r'} \cap \New{I^s} = \nbody{r'} \cap \New{I^{s-1}} = \emptyset$. Consequently, we have $r' \in \ins{r}{\New{I^s}}$, so $F \in \New{I^s}$ holds, as required.
\end{proof}

\begin{claim} \label{claiminewininsertion}
The $\supseteq$ direction of property~\eqref{insertionproperty} holds.
\end{claim}
\begin{proof}
We show by induction that \eqref{eq:fbf:proof:INewInInsertion:i} holds for each $i$.
\begin{align}
     (\Old{I^s} \setminus D^s) \cup A^s \supseteq \New{I^s_i} \label{eq:fbf:proof:INewInInsertion:i}
\end{align}

For the base case, we have $\New{I^s_0} = \New{I^{s-1}} \cup (\New{E} \cap \Out{s}) = (\Old{I^{s-1}} \setminus D^{s-1}) \cup A^{s-1} \cup (\New{E} \cap \Out{s})$ by the induction assumption for~\eqref{insertionproperty}. $\Old{I^{s-1}} \setminus D^{s-1} \subseteq \Old{I^s} \setminus D^s$ and $A^{s-1} \subseteq A^s$ clearly hold. Now consider arbitrary $F \in \New{E} \cap \Out{s}$. If $F \in E^+$, then lines~\ref{alg:dredcnt:ins:incnrloop}, \ref{alg:dredcnt:ins:incnrloop:end}, \ref{alg:dredcnt:ins:DeltaA}, and \ref{alg:dredcnt:ins:A-update} ensure $F \in (\Old{I^s} \setminus D^s) \cup A^s$. If $F \in \New{E} \setminus E^+ = \Old{E} \setminus E^-$, then we clearly have $F \in \Old{I^s}$; but then, $\Old{\Cnr[F]} = \mathsf{Occ}(F, \Old{E} \oplus \mapply{\nrstrat{\Pi}{s}}{\Old{I}})$ and property~\eqref{nonrecursivecntdecrement} imply $\Del{\Cnr[F]} \geq 1$; thus, line~\ref{alg:dredcnt:del:DeltaD} ensures $F \not\in D^s$.

For the inductive step, assume that $\New{I^s_{i-1}}$ satisfies \eqref{eq:fbf:proof:INewInInsertion:i} for $i > 0$, and consider arbitrary $F \in \New{I^s_i}$. If $F \in \New{I^s_{i-1}}$, then \eqref{eq:fbf:proof:INewInInsertion:i} holds by the induction assumption. Otherwise, by definition a rule $r$ and its instance $r' \in \ins{r}{\New{I^s_{i-1}}}$ exist where $\head{r'} = F$. Definition~\eqref{eq:ruleinstancepn} and the induction assumption for \eqref{eq:fbf:proof:INewInInsertion:i} ensure $\pbody{r'} \subseteq \New{I^s_{i-1}} \subseteq (\Old{I^{s}} \setminus D^s) \cup A^s$. Moreover, definition~\eqref{eq:ruleinstancepn} also ensures $\nbody{r'} \cap \New{I^s_{i-1}} = \emptyset$, which together with the induction assumption that property~\eqref{insertionproperty} holds for $s-1$ and the definition of $\New{I^s_{i-1}}$ implies $\nbody{r'} \cap ((\Old{I^{s-1}} \setminus D^{s-1}) \cup A^{s-1}) = \emptyset$; but then, $\nbody{r'} \subseteq \Out{<s}$ implies $\nbody{r'} \cap ((\Old{I^{s}} \setminus D^{s}) \cup A^{s}) = \emptyset$. Now we consider the following cases.
\begin{itemize}
    \item $\pbody{r'} \cap (A^s \setminus A^{s-1}) \neq \emptyset$. Facts in $A^s \setminus A^{s-1}$ are added to $A$ via lines~\ref{alg:dredcnt:ins:DeltaA} and \ref{alg:dredcnt:ins:A-update}, so there is a point in the execution of the algorithm where $\pbody{r'} \cap (A^s \setminus A^{s-1}) \cap \Delta_A \neq \emptyset$ holds in line~\ref{alg:dredcnt:ins:incrsecondloop} for the last time for $\Delta_A$. Since $\Delta_A \subseteq A$ holds at this point, we clearly have $\pbody{r'} \subseteq (\Old{I^s} \setminus D^s) \cup A$; moreover, $A \subseteq A^s$ and $\nbody{r'} \cap ((\Old{I^{s}} \setminus D^{s}) \cup A^{s}) = \emptyset$ ensure $\nbody{r'} \cap ((\Old{I^{s}} \setminus D^{s}) \cup A) = \emptyset$. But then, $r' \in \ins{r}{(\Old{I^s} \setminus D^s) \cup A \appargs \Delta_A}$ holds; $F = \head{r'}$ will be added to $N_A$ in line~\ref{alg:dredcnt:ins:NA-update}; and lines~\ref{alg:dredcnt:ins:DeltaA} and \ref{alg:dredcnt:ins:A-update} ensure $F \in (\Old{I^s} \setminus D^s) \cup A^s$.

    \item $\pbody{r'} \cap (A^s \setminus A^{s-1}) = \emptyset$, so $\pbody{r'} \subseteq (\Old{I^s} \setminus D^s) \cup A^{s-1}$ holds. We have the following two possibilities.
    \begin{itemize}
        \item $\pbody{r'} \cap (A^{s-1} \setminus D^s)$ or $\nbody{r'} \cap (D^s \setminus A^{s-1})$. But then, by the definition~\eqref{eq:ruleinstancepn} of rule matching, we clearly have ${r' \in \ins{r}{(\Old{I^s} \setminus D^s) \cup A^{s-1} \appargs A^{s-1} \setminus D^s, D^s \setminus A^{s-1}}}$. But then, lines~\ref{alg:dredcnt:ins:incnrloop}--\ref{alg:dredcnt:ins:incrloop:end} ensure that $F = \head{r'}$ is added to $N_A$; and lines~\ref{alg:dredcnt:ins:DeltaA} and \ref{alg:dredcnt:ins:A-update} ensure $F \in (\Old{I^s} \setminus D^s) \cup A^s$.

        \item $\pbody{r'} \cap (A^{s-1} \setminus D^s) = \nbody{r'} \cap (D^s \setminus A^{s-1}) = \emptyset$. But then, $\nbody{r'} \cap (D^s \setminus A^{s-1}) = \emptyset$ and $\nbody{r'} \cap (\Old{I^s} \setminus D^s) \cap A^s = \emptyset$ imply $\nbody{r'} \cap (\Old{I^s} \cup A^{s-1}) = \nbody{r'} \cap (\Old{I^s} \cup A^s) = \emptyset$. Moreover, we argue that each fact $G \in \pbody{r'} \subseteq (\Old{I^s} \setminus D^s) \cup A^{s-1}$ satisfies $G \in \Old{I^s} \setminus (D^s \setminus A^{s-1})$. This clearly holds if $G \in \Old{I^s} \setminus D^s$. If $G \in A^{s-1}$, then $G \not \in A^{s-1} \setminus D^s$ implies $G \in D^s$, which in turn implies $G \in \Old{I^s}$ by claim~\ref{claimrighthanddeletionproperty}; thus, $G \in \Old{I^s} \setminus (D^s \setminus A^{s-1})$ holds. Now, definition~\eqref{eq:ruleinstancepn} ensures $r' \in \ins{r}{\Old{I^s} \setminus (D^s \setminus A^{s-1}), \Old{I^s} \cup A^{s-1}}$, which implies $F = \head{r'} \in \Old{I^s}$. If $F \not\in D^s$, then $F \in (\Old{I^s} \setminus D^s) \cup A^s$ trivially holds. If $F \in D^s$, property~\eqref{onesteprederivationproperty} implies $F \in R^s$; then, lines~\ref{alg:dredcnt:ins:NAequalsR}, \ref{alg:dredcnt:ins:DeltaA}, and \ref{alg:dredcnt:ins:A-update} ensure $F \in (\Old{I^s} \setminus D^s) \cup A^s$.
    \end{itemize}
\end{itemize}
\end{proof}

\begin{claim} \label{claimforidaequals}
The following two properties hold.
\begin{align}
&\begin{array}{@{}l@{\;}l@{\;}l@{}} \label{nonrecursiveidaequals}
\mapply{\nrstrat{\Pi}{s}}{(\Old{I^s} \setminus D^s) \cup A^s} & = & \mapply{\nrstrat{\Pi}{s}}{\Old{I^s} \setminus (D^s \setminus A^{s-1}), \Old{I^s} \cup A^{s-1}} \\
& \oplus & \mapply{\nrstrat{\Pi}{s}}{(\Old{I^s} \setminus D^s) \cup A^s \appargs A^s \setminus D^{s-1}, D^{s-1} \setminus A^{s-1}}
\end{array} \\
&\begin{array}{@{}l@{\;}l@{\;}l@{}} \label{recursiveidaequals}
\mapply{\rstrat{\Pi}{s}}{(\Old{I^s} \setminus D^s) \cup A^s} & = & \mapply{\rstrat{\Pi}{s}}{\Old{I^s} \setminus (D^s \setminus A^{s-1}), \Old{I^s} \cup A^{s-1}} \\
& \oplus & \mapply{\rstrat{\Pi}{s}}{(\Old{I^s} \setminus D^s) \cup A^s \appargs A^s \setminus D^{s-1}, D^{s-1} \setminus A^{s-1}}
\end{array}
\end{align}
\end{claim}
\begin{proof}
By the definition of $\mapply{\Pi}{\ipos, \ineg \appargs P, N}$ it is sufficient to show that for each rule $r \in \strat{\Pi}{s}$, properties~\eqref{instidanonrepetition} and \eqref{instidaequals} hold.
\begin{align}
&\begin{array}{@{}l@{\;}l@{\;}l@{}}
\ins{r}{\Old{I^s} \setminus (D^s \setminus A^{s-1}), \Old{I^s} \cup A^{s-1}} \cap \ins{r}{(\Old{I^s} \setminus D^s) \cup A^s \appargs A^s \setminus D^{s-1}, D^{s-1} \setminus A^{s-1}} = \emptyset
\end{array} & \label{instidanonrepetition} \\
&\begin{array}{@{}l@{\;}l@{\;}l@{}}
\ins{r}{(\Old{I^s} \setminus D^s) \cup A^s} & = & \ins{r}{\Old{I^s} \setminus (D^s \setminus A^{s-1}), \Old{I^s} \cup A^{s-1}} \\
& \cup & \ins{r}{(\Old{I^s} \setminus D^s) \cup A^s \appargs A^s \setminus D^{s-1}, D^{s-1} \setminus A^{s-1}}
\end{array} & \label{instidaequals}
\end{align}
To prove \eqref{instidanonrepetition}, consider an arbitrary rule instance $r' \in \ins{r}{(\Old{I^s} \setminus D^s) \cup A^s \appargs A^s \setminus D^{s-1}, D^{s-1} \setminus A^{s-1}}$. By definition~\eqref{eq:ruleinstancepn} we have $\pbody{r'} \cap (A^s \setminus D^{s-1}) \neq \emptyset$ or $\nbody{r'} \cap (D^{s-1} \setminus A^{s-1}) \neq \emptyset$; now we examine these two cases separately.

For the first case, let $F$ be an arbitrary fact in $\pbody{r'} \cap (A^s \setminus D^{s-1})$. Now if $F \in A^{s-1} \setminus D^{s-1}$ holds, then the induction assumption for \eqref{ioldcapaind} ensures $F \not\in \Old{I^{s-1}}$, which in turn implies $F \not \in \Old{I^s} \setminus (D^s \setminus A^{s-1})$; if $F \in A^s \setminus A^{s-1}$ holds, then in the same way as in the proof of claim~\ref{claimforioldcapaind} we have $F \not \in (\Old{I^s} \setminus D^s) \cup A^{s-1}$, which implies $F \not \in \Old{I^s} \setminus (D^s \setminus A^{s-1})$; either way, we have $\pbody{r'} \not\subseteq \Old{I^s} \setminus (D^s \setminus A^{s-1})$, so $r' \not \in \ins{r}{\Old{I^s} \setminus (D^s \setminus A^{s-1}), \Old{I^s} \cup A^{s-1}}$ holds.

For the second case, $\nbody{r'} \cap (D^{s-1} \setminus A^{s-1}) \neq \emptyset$ and the induction assumption for property~\eqref{deletionproperty} imply $\nbody{r'} \cap \Old{I^{s-1}} \neq \emptyset$, which clearly implies $\nbody{r'} \cap (\Old{I^s} \cup A^{s-1}) \neq \emptyset$; thus, by definition~\eqref{eq:ruleinstancepn} we have $r' \not \in \ins{r}{\Old{I^s} \setminus (D^s \setminus A^{s-1}), \Old{I^s} \cup A^{s-1}}$; this completes our proof for \eqref{instidanonrepetition}.

Next we prove the $\supseteq$ direction of property~\eqref{instidaequals}. Consider an arbitrary rule instance $r'$ contained in the right-hand side of \eqref{instidaequals}: if $r' \in \ins{r}{(\Old{I^s} \setminus D^s) \cup A^s \appargs A^s \setminus D^{s-1}, D^{s-1} \setminus A^{s-1}}$ holds, then by definition~\eqref{eq:ruleinstancepn} we clearly have $r' \in \ins{r}{(\Old{I^s} \setminus D^s) \cup A^s}$; if we have $r' \in \ins{r}{\Old{I^s} \setminus (D^s \setminus A^{s-1}), \Old{I^s} \cup A^{s-1}}$, then $\pbody{r'} \subseteq \Old{I^s} \setminus (D^s \setminus A^{s-1})$ implies $\pbody{r'} \subseteq (\Old{I^s} \setminus D^s) \cup A^s$; moreover, $\nbody{r'} \cap (\Old{I^s} \cup A^{s-1}) = \emptyset$ and $\nbody{r'} \subseteq \Out{<s}$ jointly imply $\nbody{r'} \cap ((\Old{I^s} \setminus D^s) \cup A^s) = \nbody{r'} \cap ((\Old{I^{s-1}} \setminus D^{s-1}) \cup A^{s-1}) = \emptyset$; therefore, $r' \in \ins{r}{(\Old{I^s} \setminus D^s) \cup A^s}$ holds by definition~\eqref{eq:ruleinstancepn}.

Finally, for the $\subseteq$ direction of property~\eqref{instidaequals}, consider arbitrary $r' \in \ins{r}{(\Old{I^s} \setminus D^s) \cup A^s}$. If $\pbody{r'} \cap (A^s \setminus D^{s-1}) \neq \emptyset$ or $\nbody{r'} \cap (D^{s-1} \setminus A^{s-1}) \neq \emptyset$ holds, then we clearly have $r' \in \ins{r}{(\Old{I^s} \setminus D^s) \cup A^s \appargs A^s \setminus D^{s-1}, D^{s-1} \setminus A^{s-1}}$ by definition~\eqref{eq:ruleinstancepn}. Otherwise, let $F$ be an arbitrary fact in $\pbody{r'}$ and let $G$ be an arbitrary fact in $\nbody{r'}$, then we have $F \in ((\Old{I^s} \setminus D^s) \cup A^s) \setminus (A^s \setminus D^{s-1})$ and $G \not\in (\Old{I^s} \setminus D^s) \cup A^s \cup (D^{s-1} \setminus A^{s-1})$. We next show that $F \in \Old{I^s} \setminus (D^s \setminus A^{s-1})$ and $G \not\in \Old{I^s} \cup A^{s-1}$ hold.

If $F \in (\Old{I^s} \setminus D^s) \setminus (A^s \setminus D^{s-1})$ holds, then $F \in \Old{I^s} \setminus (D^s \setminus A^{s-1})$ trivially holds; if $F \in A^s \setminus (A^s \setminus D^{s-1})$ holds, then we have $F \in A^s \cap D^{s-1} = A^{s-1} \cap D^{s-1} = A^{s-1} \cap D^s$, which in turn implies $F \not \in D^s \setminus A^{s-1}$; moreover, $F \in D^s$ and property~\eqref{deletionproperty} imply $F \in \Old{I^s}$; therefore, $F \in \Old{I^s} \setminus (D^s \setminus A^{s-1})$ holds, as required.

$G \not \in A^s$ and $G \not \in D^{s-1} \setminus A^{s-1}$ jointly imply $G \not\in D^{s-1}$; but then, $G \in \Out{<s}$ implies $G \not \in D^s$, which together with $G \not \in \Old{I^s} \setminus D^s$ ensures $G \not \in \Old{I^s}$; thus, $G \not \in \Old{I^s} \cup A^{s-1}$ holds, as required.

$F$ and $G$ are chosen arbitrarily, so we have $\pbody{r'} \in \Old{I^s} \setminus (D^s \setminus A^{s-1})$ and $\nbody{r'} \not\in \Old{I^s} \cup A^{s-1}$; by definition~\eqref{eq:ruleinstancepn} we have $r' \in \ins{r}{\Old{I^s} \setminus (D^s \setminus A^{s-1}), \Old{I^s} \cup A^{s-1}}$. This completes our proof for \eqref{instidaequals}.
\end{proof}

\begin{claim}
Property~\eqref{nonrecursivecntincrement} holds.
\end{claim}
\begin{proof}
Due to claim~\ref{claimforidaequals} it is now sufficient to show that for each $F \in ((\Old{I^s} \setminus D^s) \cup A^s) \cap \Old{I^s} \cap \Out{s}$, the following property holds.
\begin{align}
\Add{\Cnr[F]}-\Del{\Cnr[F]} = \mathsf{Occ}(F, (E^+ \cap \Out{s}) \oplus \mapply{\nrstrat{\Pi}{s}}{(\Old{I^s} \setminus D^s) \cup A^s \appargs A^s \setminus D^{s-1}, D^{s-1} \setminus A^{s-1}})
\end{align}
$\nrstrat{\Pi}{s}$ contains only nonrecursive rules, so it is sufficient to prove
\begin{align}
\Add{\Cnr[F]}-\Del{\Cnr[F]} = \mathsf{Occ}(F, (E^+ \cap \Out{s}) \oplus \mapply{\nrstrat{\Pi}{s}}{(\Old{I^s} \setminus D^s) \cup A^{s-1} \appargs A^{s-1} \setminus D^{s-1}, D^{s-1} \setminus A^{s-1}}), \label{nonrecursivecntincrementequivalent}
\end{align}
which is ensured by line~\ref{alg:dredcnt:ins:incnrloop} and line~\ref{alg:dredcnt:ins:incnrloop:end} of the algorithm.
\end{proof}

\begin{claim}
Property~\eqref{nonrecursivenew} holds.
\end{claim}
\begin{proof}
$F \not \in \Old{I^s} \cap \Out{s}$ ensures $\Del{\Cnr[F]} = \Old{\Cnr[F]} = 0$. Moreover, $F \not \in \mapply{\nrstrat{\Pi}{s}}{\Old{I^s}}$ and property~\eqref{nonrecursiveoldequals} jointly imply $F \not \in \mapply{\nrstrat{\Pi}{s}}{\Old{I^s} \setminus (D^s \setminus A^{s-1}), \Old{I^s} \cup A^{s-1}}$. Now if ${F \not\in (E^+ \cap \Out{s}) \oplus \mapply{\nrstrat{\Pi}{s}}{(\Old{I^s} \setminus D^s) \cup A^{s-1} \appargs A^{s-1} \setminus D^{s-1}, D^{s-1} \setminus A^{s-1}}}$, then   by property~\eqref{nonrecursiveidaequals} and $F \not \in (\Old{E} \setminus E^-) \cap \Out{s}$ we clearly have $\mathsf{Occ}(F, ((\Old{E} \setminus E^-) \cup E^+) \oplus \mapply{\nrstrat{\Pi}{s}}{(\Old{I^s} \setminus D^s) \cup A^s}) = 0$; moreover, line~\ref{alg:dredcnt:ins:incnrloop} ensures that the nonrecursive count for $F$ is not incremented in this case, so the left-hand side of \eqref{nonrecursivenew} equals zero as well. Otherwise, each occurrence of $F$ in the multiset $(E^+ \cap \Out{s}) \oplus \mapply{\nrstrat{\Pi}{s}}{(\Old{I^s} \setminus D^s) \cup A^{s-1} \appargs A^{s-1} \setminus D^{s-1}, D^{s-1} \setminus A^{s-1}}$ results in incrementing the nonrecursive count of $F$ by one; together with $F \not \in (\Old{E} \setminus E^-) \cap \Out{s}$ and $F \not \in \mapply{\nrstrat{\Pi}{s}}{\Old{I^s} \setminus (D^s \setminus A^{s-1}), \Old{I^s} \cup A^{s-1}}$ this ensures the correctness of the property.
\end{proof}

\begin{claim}
Property~\eqref{recursivecntincrement} holds.
\end{claim}
\begin{proof}
Line~\ref{alg:dredcnt:ins:DeltaA} and line~\ref{alg:dredcnt:ins:A-update} ensure that $\Delta_A$ used in line~\ref{alg:dredcnt:ins:incrsecondloop} is different between iterations of the loop in lines~\ref{alg:dredcnt:ins:loop}--\ref{alg:dredcnt:ins:loop:end}, so the rle instances considered in line~\ref{alg:dredcnt:ins:incrsecondloop} are different between iterations. All these rule instances are in $\bigcup_{r' \in \rstrat{\Pi}{s}}\ins{r'}{(\Old{I^s} \setminus D^s) \cup A^s}$, which is equal to $\bigcup_{r' \in \rstrat{\Pi}{s}}\ins{r'}{\New{I^s}}$ by property~\eqref{insertionproperty}. Therefore, the loop will terminate. Now let $T$ be the total number of iterations; moreover, for each $1 \leq i \leq T$, let $A^s_i$ be the value of $A$ at the beginning of the $i$th iteration of the loop, and let $\Add{\Cr[F]}^i$ be the value of $\Cr[F]$ for each $F \in ((\Old{I^s} \setminus D^s) \cup A^s) \cap \Old{I^s} \cap \Out{s}$ at the same time point. We next prove by induction on $i$ that~\eqref{deltainccntinduction} holds for $1 \leq i \leq T$; then~\eqref{deltainccntinduction} for $i=T$ and property~\eqref{recursiveidaequals} ensure the correctness of the claim.
\begin{align}
\Add{\Cr[F]}^i-\Del{\Cr[F]} = \mathsf{Occ}(F, \mapply{\rstrat{\Pi}{s}}{(\Old{I^s} \setminus D^s) \cup A^s_i \appargs A^s_i \setminus D^{s-1}, D^{s-1} \setminus A^s_i}) \label{deltainccntinduction}
\end{align}

For the base case, we have $A^s_1 = A^{s-1}$. Consider arbitrary $F \in ((\Old{I^s} \setminus D^s) \cup A^s) \cap \Old{I^s} \cap \Out{s}$. Lines~\ref{alg:dredcnt:ins:incrloop} and \ref{alg:dredcnt:ins:incrloop:end} ensure that each occurrence of $F$ in $\mapply{\rstrat{\Pi}{s}}{(\Old{I^s} \setminus D^s) \cup A^{s-1} \appargs A^{s-1} \setminus D^s, D^s \setminus A^{s-1}}$ results in incrementing the corresponding recursive count by one. But then, stratification ensures $\mapply{\rstrat{\Pi}{s}}{(\Old{I^s} \setminus D^s) \cup A^{s-1} \appargs A^{s-1} \setminus D^s, D^s \setminus A^{s-1}} = \mapply{\rstrat{\Pi}{s}}{(\Old{I^s} \setminus D^s) \cup A^{s-1} \appargs A^{s-1} \setminus D^{s-1}, D^{s-1} \setminus A^{s-1}}$, so \eqref{deltainccntinduction} holds for $i=1$, as required.

For the inductive step, assume that~\eqref{deltainccntinduction} holds for $i-1$ where $1 < i \leq T$, and consider arbitrary $F \in ((\Old{I^s} \setminus D^s) \cup A^s) \cap \Old{I^s} \cap \Out{s}$. Lines~\ref{alg:dredcnt:ins:DeltaA}, \ref{alg:dredcnt:ins:A-update}, and \ref{alg:dredcnt:ins:incrsecondloop} jointly imply~\eqref{recursivecntincrementinduction}.
\begin{align}
\Add{\Cr[F]}^i - \Add{\Cr[F]}^{i-1} = \mathsf{Occ}(F, \mapply{\rstrat{\Pi}{s}}{(\Old{I^s} \setminus D^s) \cup A^s_i \appargs A^s_i \setminus A^s_{i-1}}) \label{recursivecntincrementinduction}
\end{align}
We now show that the following holds.
\begin{align}
\begin{array}{@{}r@{\;}r@{\;}r@{}} \label{recursivecntincrementoccurrence}
\mathsf{Occ}(F, \mapply{\rstrat{\Pi}{s}}{(\Old{I^s} \setminus D^s) \cup A^s_{i-1} \appargs A^s_{i-1} \setminus D^{s-1}, D^{s-1} \setminus A^s_{i-1}} + \mathsf{Occ}(F, \mapply{\rstrat{\Pi}{s}}{(\Old{I^s} \setminus D^s) \cup A^s_i \appargs A^s_i \setminus A^s_{i-1}})& \\
= \mathsf{Occ}(F, \mapply{\rstrat{\Pi}{s}}{(\Old{I^s} \setminus D^s) \cup A^s_i \appargs A^s_i \setminus D^{s-1}, D^{s-1} \setminus A^s_i}) &
\end{array}
\end{align}
To this end, it is sufficient to show that for each $r \in \rstrat{\Pi}{s}$, property~\eqref{recursivecntincrementinductioninstancelevel} holds.
\begin{align}
\begin{array}{@{}r@{\;}r@{\;}r@{}} \label{recursivecntincrementinductioninstancelevel}
\ins{r}{(\Old{I^s} \setminus D^s) \cup A^s_{i-1} \appargs A^s_{i-1} \setminus D^{s-1}, D^{s-1} \setminus A^s_{i-1}} \cup \ins{r}{(\Old{I^s} \setminus D^s) \cup A^s_i \appargs A^s_i \setminus A^s_{i-1}} & \\
= \ins{r}{(\Old{I^s} \setminus D^s) \cup A^s_i \appargs A^s_i \setminus D^{s-1}, D^{s-1} \setminus A^s_i}
\end{array}
\end{align}
$(A^s_{i-1} \setminus D^{s-1}) \cup (A^s_i \setminus A^s_{i-1}) = A^s_i \setminus D^{s-1}$ and $D^{s-1} \setminus A^s_{i-1} = D^{s-1} \setminus A^s_i \subseteq \Out{<s}$ ensure that the $\subseteq$ direction of \eqref{recursivecntincrementinductioninstancelevel} holds. To see that the $\supseteq$ direction holds as well, consider arbitrary $r' \in \ins{r}{(\Old{I^s} \setminus D^s) \cup A^s_i \appargs A^s_i \setminus D^{s-1}, D^{s-1} \setminus A^s_i}$. If $\pbody{r'} \cap (A^s \setminus A^{s-1}) \neq \emptyset$, then we clearly have $r' \in \ins{r}{(\Old{I^s} \setminus D^s) \cup A^s_i \appargs A^s_i \setminus A^s_{i-1}}$. If $\pbody{r'} \cap (A^s_i \setminus A^s_{i-1}) = \emptyset$, then definition~\eqref{eq:ruleinstancepn} ensures that we have $\pbody{r'} \subseteq ((\Old{I^s} \setminus D^s) \cup A^s_i) \setminus (A^s_i \setminus A^s_{i-1}) = (\Old{I^s} \setminus D^s) \cup A^s_{i-1}$. There are two possibilities here: if $\pbody{r'} \cap (A^s_i \setminus D^{s-1}) \neq \emptyset$, then $\pbody{r'} \cap (A^s_i \setminus A^s_{i-1}) = \emptyset$ implies $\pbody{r'} \cap (A^s_{i-1} \setminus D^{s-1}) \neq \emptyset$; if $\nbody{r'} \cap (D^{s-1} \setminus A^s_i) \neq \emptyset$, then clearly $\nbody{r'} \cap (D^{s-1} \setminus A^s_{i-1}) \neq \emptyset$ holds as well; either way, we have $F \in \ins{r}{(\Old{I^s} \setminus D^s) \cup A^s_{i-1} \appargs A^s_{i-1} \setminus D^{s-1}, D^{s-1} \setminus A^s_{i-1}}$. Therefore, the $\supseteq$ direction of property~\eqref{recursivecntincrementinductioninstancelevel} holds. But then, property~\eqref{recursivecntincrementoccurrence} holds, which together with \eqref{recursivecntincrementinduction} and the induction assumption for~\eqref{deltainccntinduction} ensures that~\eqref{deltainccntinduction} holds for $i$ as well.
\end{proof}

\begin{claim}
Property~\eqref{recursivenew} holds.
\end{claim}
\begin{proof}
$F \not \in \Old{I^s}$ implies $\Del{\Cr[F]} = 0$, so we have $\Add{\Cr[F]} = \mathsf{Occ}(F, \mapply{\rstrat{\Pi}{s}}{(\Old{I^s} \setminus D^s) \cup A^s \appargs A^s \setminus D^{s-1}, D^{s-1} \setminus A^{s-1}})$ in the same way as in the proof for the previous claim. Moreover, $F \not\in \Old{I^s}$ ensures $\mathsf{Occ}(F, \mapply{\rstrat{\Pi}{s}}{\Old{I^s} \setminus (D^s \setminus A^{s-1}), \Old{I^s} \cup A^{s-1}}) = 0$. Therefore, by property~\eqref{recursiveidaequals} we have $\Add{\Cr[F]}  = \mathsf{Occ}(F, \mapply{\rstrat{\Pi}{s}}{(\Old{I^s} \setminus D^s) \cup A^s})$, as required.
\end{proof}

}{}

\end{document}